\documentclass[preprint, 3p]{elsarticle}

\usepackage[english]{babel}
\usepackage[utf8]{inputenc}
\usepackage{amsthm}
\usepackage{amsmath}
\usepackage{amssymb}
\usepackage{hyperref}
\usepackage{enumerate}
\usepackage{xcolor}
\usepackage{color,soul}
\usepackage{tikz}
\usepackage{placeins}
\usepackage{subfigure}
\usepackage{booktabs}
\usepackage{bbm}
\usepackage{bm}
\usepackage{dsfont}
\graphicspath{{images/}{../images/}}
\usepackage{eurosym}

\newcommand{\Rr}{{\mathbb{R}}}

\newcommand{\Aa}{{\mathcal{A}}}

\def\leq{\leqslant}

\newtheorem{theorem}{Theorem}

\newtheorem{proposition}{Proposition}

\newtheorem{proposition*}{Proposition}
\newtheorem{remark}{Remark}

\begin{document}
	
\begin{frontmatter}
	
\title{Closed-form approximations in multi-asset market making}

\author[UParis1]{Philippe Bergault}
\ead{philippe.bergault@etu.univ-paris1.fr}
\address[UParis1]{Universit\'e  Paris 1 Panthéon-Sorbonne, Centre d'Economie de la Sorbonne, 106 Boulevard de l'H\^opital, 75642 Paris Cedex~13, France.}

\author[FGV]{David Evangelista}
\ead{david.evangelista@fgv.br}

\address[FGV]{Escola de Matem\'atica Aplicada, Funda\c c\~ao Get\'ulio Vargas, Rio de Janeiro, RJ, Brasil.}

\author[UParis1]{Olivier Gu\'eant}
\ead{olivier.gueant@univ-paris1.fr}

\author[ICL]{Douglas Vieira}
\ead{d.machado-vieira15@imperial.ac.uk}
\address[ICL]{Imperial College London, Department of Mathematics,
London SW7 2AZ, UK.}

\date{}
\journal{ }
	
	
\begin{abstract}
A large proportion of market making models derive from the seminal model of Avellaneda and Stoikov. The numerical approximation of the value function and the optimal quotes in these models remains a challenge when the number of assets is large. In this article, we propose closed-form approximations for the value functions of many multi-asset extensions of the Avellaneda-Stoikov model. These approximations or proxies can be used (i)~as heuristic evaluation functions, (ii)~as initial value functions in reinforcement learning algorithms, and/or (iii)~directly to design quoting strategies through a greedy approach. Regarding the latter, our results lead to new and easily interpretable closed-form approximations for the optimal quotes, both in the finite-horizon case and in the asymptotic (ergodic) regime.
\end{abstract}


\end{frontmatter}	
\setlength\parindent{0pt}
\section{Introduction}

Since the publication of the paper \cite{avellaneda2008high} by Avellaneda and Stoikov, who revisited the paper \cite{ho1981optimal} by Ho and Stoll (see also \cite{ho1983dynamics}), there has been an extensive literature on optimal market making.\footnote{There is an economic literature on market making, for instance the seminal paper \cite{grossman1988liquidity} by Grossman and Miller. The results in this literature are, however, more interesting for understanding the price formation process than for building market making algorithms.} Guéant, Lehalle, and Fernandez-Tapia provided in \cite{gueant2013dealing} a rigorous analysis of the stochastic optimal control problem introduced by Avellaneda and Stoikov and proved that, under inventory constraints, the problem reduces to a system of linear ordinary differential equations in the case of exponential intensity functions suggested by Avellaneda and Stoikov. They also studied the asymptotics when the time horizon $T$ tends to $+\infty$, proposed closed-form approximations, and introduced extensions to include a drift in the price dynamics and market impact~/~adverse selection. Cartea and Jaimungal, along with their various coauthors, contributed substantially to the literature and added many features to the initial models:  alpha signals, ambiguity aversion, etc. (see~\cite{cartea2017algorithmic,cartea2014buy, cartea2018algorithmic} -- see also their book \cite{cartea2015algorithmic}). They also considered a different objective function: the expected PnL minus a running penalty to avoid holding a large inventory instead of the Von Neumann-Morgenstern expected CARA (constant absolute risk aversion) utility of \cite{avellaneda2008high} and \cite{gueant2013dealing}. Many features have also been added by various authors: general dynamics for the price in \cite{fodra2012high}, general intensities and partial information \cite{campi2020optimal}, persistence of the order flow in \cite{jusselin2020optimal}, several requested sizes in \cite{bergault2019size}, client tiering and access to a liquidity pool in \cite{barzykin2020algorithmic}, etc.\\

In spite of the focus of initial papers on stock markets,\footnote{There was also from the very beginning a focus on options markets -- see for instance \cite{stoikov2009option} (cf. \cite{baldacci2019algorithmic} and \cite{el2015stochastic} for more recent papers).} the models derived from that of Avellaneda and Stoikov have been more useful to build market making algorithms in quote-driven markets: corporate bond markets based on requests for quotes, FX markets based on requests for quotes and requests for stream, etc. For stock markets or, more generally, order-driven markets with relatively low bid-ask spread to tick size ratio, many models have been proposed that depart from the original framework of Avellaneda and Stoikov in that the limit order book is modeled. Instances of papers proposing this type of models include those of Guilbaud and Pham \cite{guilbaud2013optimal, guilbaud2015optimal}, that of K\"uhn and Muhle-Karbe \cite{kuhn2013optimal}, that of Fodra and Pham \cite{fodra2015high} or the more recent papers by Lu and Abergel \cite{lu2018order} and Baradel, Bouchard, Evangelista, and Mounjid \cite{baradel2018optimal}.\\

Most of the literature on optimal market making deals with single-asset models. However, because market making algorithms are typically built for entire portfolios, single-asset models are not sufficient to build operable algorithms, except under the unrealistic assumption that asset prices are uncorrelated. Multi-asset extensions of the Avellaneda-Stoikov model have been proposed. A paper by Gu\'eant and Lehalle \cite{gueant2015general} touches upon this extension and a complete analysis for the various objective functions present in the literature can be found in \cite{gueant2017optimal} (see also the book \cite{gueant2016financial}) or in \cite{bergault2019size} in which multiple trade sizes are also considered.\\

Although their mathematical characterization has been known for years, computing the value function and the optimal quotes is complicated in the multi-asset case whenever the prices of the assets are correlated. The grid methods that are classically used to tackle the single-asset case suffer indeed from the curse of dimensionality and do not scale up to many practical multi-asset cases. Bergault and Guéant proposed in~\cite{bergault2019size} a factor method to reduce the dimensionality of the problem. Guéant and Manziuk proposed in \cite{gueant2019deep} a numerical method based on reinforcement learning techniques (an actor-critic approach in fact). In spite of these recent advances, the computational cost of most numerical schemes will still be prohibitive for practical use for some asset classes.\\

Instead of computing a numerical approximation of the value function (from which one traditionally deduces a numerical approximation of the optimal quotes), we propose in this paper a method for building a closed-form proxy for the value function. The idea behind the approach is that the value function associated with many market making problems is the solution of a Hamilton-Jacobi equation that can be ``approximated'' by another Hamilton-Jacobi equation for which the solution can be computed in closed form. Of course,  such closed-form formula does not define a solution to the initial Hamilton-Jacobi equation, but it has similar properties and should capture most of the relevant financial effects.\\

Having a proxy of a value function is known to be useful in the community of reinforcement learning (see \cite{sutton} and \cite{cs} for a reference to the reinforcement learning terminology). An important use of a closed-form proxy of a value function is as a heuristic evaluation function. Heuristic evaluation functions are mainly used in game-playing computer programs to evaluate the probability to win the game given the current state -- usually the current board in board games -- but they can be used as terminal values in many Monte-Carlo-based reinforcement learning techniques. Also, such a proxy can be used as a starting point for many iterative algorithms based on value functions: value iteration algorithm, actor-critic approaches, etc. The last application we highlight -- which was also our initial motivation -- is that one can build from a proxy of a value function a quoting strategy by using what is called in the reinforcement literature the greedy strategy associated with that proxy (i.e. the strategy that makes the locally optimal choice if at each time step the value function associated with the tail problem is replaced by its proxy in the dynamic programming equation). Having such a strategy in closed form has numerous advantages. First, it can be used directly by market practitioners as a quoting strategy. Second, it can be used as a starting point in iterative algorithms based on policy functions: policy iteration algorithm, actor-critic approaches, etc. Third, it has the advantage of being easily interpretable and gives insights on the true optimal strategy such as the identification of the leading factors and the sensitivity to changes in model parameters.\footnote{For market making, the influence of the parameters has already been studied in \cite{gueant2013dealing} (one-asset case) and \cite{gueant2017optimal} (multi-asset case).}\\


The method we propose is first applied to the multi-asset market making models of \cite{gueant2017optimal}. Then we generalize the framework in several directions to cover many important practical cases: (i)~drift in prices, (ii)~client tiering, (iii)~several request sizes for each asset and each tier, and (iv)~fixed transaction costs for each asset and each tier. The drift in prices models the views of the market maker. Client tiering is a common practice in OTC markets, justified by the large spectrum of needs and behaviors in the set of clients to be served. The introduction of several request sizes for each asset and each tier reflects the reality that request sizes are not in control of the market makers, but rather of their clients. The fixed transaction costs can model extra costs associated with the market making business, for instance related to trading platforms.\\

We end this introduction by outlining our paper. In Section~\ref{sec:systemODE} we recall the multi-asset extensions of the Avellaneda-Stoikov model proposed in \cite{gueant2017optimal}, present the system of ordinary differential equations (the Hamilton-Jacobi equation) characterizing the value function, and state the main results regarding the optimal quotes. In Section~\ref{sec:generaleq}, we present our approach and compute a closed-form proxy for the value function. We deduce from that proxy an approximation of the optimal quotes in closed form. In Section~\ref{sec:perturbative}, we use a perturbation approach to propose a correction term that can easily be computed thanks to Monte-Carlo simulations. In Section~\ref{sec:extensions}, we extend our results to a more general multi-asset market making model with drift in prices, client tiering, several requested sizes for each asset and each tier, and fixed transaction costs for each asset and each tier. Numerical examples are presented in Section \ref{num_sec}. They illustrate the quality of our closed-form approximations.\\

\section{The multi-asset market making model}\label{sec:systemODE}

\subsection{Model setup}
\vspace{3mm}
\label{baseModel}

We fix a probability space $(\Omega,\mathcal F, \mathbb{P})$ equipped with a filtration $(\mathcal{F}_t)_{t\in\Rr_+}$ satisfying the usual conditions. In what follows, we assume that all stochastic processes are defined on $(\Omega,\mathcal F,(\mathcal{F}_t)_{t\in\Rr_+}, \mathbb{P})$. In all this paper, $\mathbb R_+$ denotes the set of nonnegative real numbers, and $\mathbb R_+^*$ denotes the set of positive real numbers.\\

For $i\in \{1,\ldots,d\}$, the reference price of asset $i$ is modeled by a process $(S_t^i)_{t \in \mathbb R_+}$ with dynamics
\[
dS^i_t=\sigma^idW^i_t,\quad \mbox{$S_0^i$ given,}
\]
where $(W_t^1,\ldots,W_t^d)_{t \in \mathbb R_+}$ is a $d$-dimensional Brownian motion with correlation matrix $(\rho^{i,j})_{1\leq i,j\leq d}$  adapted to the filtration $(\mathcal{F}_t)_{t\in \mathbb{R}_+}$ -- hereafter we denote by $\Sigma=(\rho^{i,j}\sigma^i\sigma^j)_{1\leq i,j\leq d}$ the variance-covariance matrix associated with the process $(S_t)_{t \in \mathbb R_+}=(S_t^1,\ldots,S_t^d)_{t \in \mathbb R_+}$.\\

The market maker chooses at each point in time the price at which she is ready to buy/sell each asset: for $i\in \{1,\ldots,d\}$, we let her bid and ask quotes for asset $i$ be modeled by two stochastic processes, respectively denoted by $(S_t^{i,b})_{t \in \mathbb R_+}$  and $(S_t^{i,a})_{t \in \mathbb R_+}$.\\

For $i\in \{1,\ldots,d\}$, we denote by $(N_t^{i,b})_{t \in \mathbb R_+}$ and $(N_t^{i,a})_{t \in \mathbb R_+}$ the two point processes modeling the number of transactions at the bid and at the ask, respectively, for asset $i$. We assume in this section that the transaction size for asset $i$ is constant and denoted by $z^i$. The inventory process of the market maker for asset $i$, denoted by $(q^i_t)_{t \in \mathbb R_+}$, has therefore the dynamics
\[
dq_t^i=z^idN_t^{i,b}-z^idN_t^{i,a},\quad \mbox{$q_0^i$ given,}
\]
and we denote by $(q_t)_{t \in \mathbb R_+}$ the (column) vector process $\left(q^1_t, \ldots, q^d_t \right)^\intercal _{t \in \mathbb R_+}$.\\

For each $i \in \{1,\ldots,d\}$, we denote by $(\lambda_t^{i,b})_{t \in \mathbb R_+}$ and $(\lambda_t^{i,a})_{t \in \mathbb R_+}$ the intensity processes of $(N^{i,b}_t)_{t \in \mathbb R_+}$ and $(N^{i,a}_t)_{t \in \mathbb R_+}$, respectively. We assume that the market maker stops proposing a bid (respectively ask) price for asset  $i$ when her position in asset $i$ following the transaction would exceed a given threshold $Q^i$ (respectively~$-Q^i$).\footnote{$Q^i$ is assumed to be a multiple of $z^i$. It corresponds to the risk limit of the market maker for asset $i$.}\\

Formally, we assume that the intensities verify
$$
\lambda^{i,b}_t=\Lambda^{i,b}(\delta_t^{i,b}) \mathds{1}_{\{q^i_{t-} + z^i \le Q^i\}}\quad\mbox{and}\quad\lambda^{i,a}_t=\Lambda^{i,a}(\delta_t^{i,a})\mathds{1}_{\{q^i_{t-} - z^i \ge -Q^i\}},
$$
where the processes $(\delta^{i,b}_t)_{t \in \mathbb R_+}$ and $(\delta^{i,a}_t)_{t \in \mathbb R_+}$ are defined by\footnote{It is often assumed in the literature that the point processes are independent of the Brownian motions. In that case, the quote processes $(\delta^{i,b}_t)_{t \in \mathbb R_+}$ and $(\delta^{i,a}_t)_{t \in \mathbb R_+}$ have to be independent of prices. In fact, the optimal control problem can be written in a weak form to show that this assumption is not necessary -- see \ref{ConstrucN} for more details on the construction of the processes in that case.}
\[
\delta_t^{i,b}=S_t^{i}-S_t^{i,b}\quad \mbox{and}\quad \delta_t^{i,a}=S_t^{i,a}-S_t^{i}, \quad \forall t \in \mathbb R_+.
\]
Moreover, we assume that the functions $\Lambda^{i,b}$ and $\Lambda^{i,a}$ satisfy the following properties:

\begin{itemize}
	\item $\Lambda^{i,b}$ and $\Lambda^{i,a}$ are twice continuously differentiable,
	\item $\Lambda^{i,b}$ and $\Lambda^{i,a}$ are decreasing, with $\forall \delta \in \Rr$, ${\Lambda^{i,b}}'(\delta) <0$ and ${\Lambda^{i,a}}'(\delta) <0$,
	\item $\lim_{\delta \to +\infty} \Lambda^{i,b}(\delta) = \lim_{\delta \to +\infty} \Lambda^{i,a}(\delta) = 0$,
	\item $\sup_{\delta} \frac{\Lambda^{i,b}(\delta){\Lambda^{i,b}}''(\delta)}{ \left({\Lambda^{i,b}}'(\delta)\right)^2} < 2 \quad \text{and} \quad \sup_{\delta} \frac{\Lambda^{i,a}(\delta){\Lambda^{i,a}}''(\delta)}{\left({\Lambda^{i,a}}'(\delta)\right)^2}  < 2.$\\
\end{itemize}

Finally, the process $(X_t)_{t \in \mathbb R_+}$ modelling the amount of cash on the market maker's cash account has the following dynamics:
\begin{align*}
dX_t&=\sum\limits_{i=1}^{d} S_t^{i,a}z^i dN_t^{i,a}-S_t^{i,b}z^i dN_t^{i,b}\\
&=\sum\limits_{i=1}^{d}(S_t^i+\delta_t^{i,a})z^i dN_t^{i,a}-(S_t^i-\delta_t^{i,b})z^i dN_t^{i,b}\\
&=\sum\limits_{i=1}^{d}\left(\delta_t^{i,b}z^i dN_t^{i,b} + \delta_t^{i,a} z^i dN_t^{i,a} \right) - \sum\limits_{i=1}^{d} S^i_t dq^i_t.
\end{align*}

\subsection{The optimization problems}
\vspace{3mm}
We can consider two different optimization problems for the market maker. Following the initial model proposed by Avellaneda and Stoikov in \cite{avellaneda2008high}, we can assume that she maximizes the expected value of a CARA utility function (with risk aversion parameter $\gamma>0$) applied to the mark-to-market value of her portfolio at a given time $T$. This mark-to-market value is the sum of the amount $X_T$ on the cash account and the mark-to-market value $\sum\limits_{i=1}^{d}q_T^i S_T^i$ of the assets remaining in the portfolio at date $T$.\footnote{In the literature there is sometimes a penalty function applied to the inventory at terminal time $T$ to ``force'' liquidation. Here, as we shall focus on the asymptotic regime of the optimal quotes, there is no point considering such a penalty. However, it is noteworthy that most of our non-asymptotic results could be generalized to the case of a quadratic terminal penalty.} More precisely, her optimization problem writes
\begin{align*}
\sup_{\substack{(\delta_t^{1,b})_t,\ldots,(\delta_t^{d,b})_t\in\mathcal A\\ (\delta_t^{1,a})_t,\ldots,(\delta_t^{d,a})_t\in\mathcal A}} \mathbb E\left[-\exp \left(-\gamma \left(X_T + \sum\limits_{i=1}^{d}q_T^i S_T^i \right)\right)\right],
\end{align*}
where $\Aa$ is the set of predictable processes bounded from below. We call Model A our model with this first objective function.\\

Alternatively, as proposed by Cartea \emph{et al.} in \cite{cartea2014buy}, we can consider a risk-adjusted expectation for the objective function of the market maker. In that case, the optimization problem writes
\begin{align*}
	\sup_{\substack{(\delta_t^{1,b})_t,\ldots,(\delta_t^{d,b})_t\in\mathcal A\\ (\delta_t^{1,a})_t,\ldots,(\delta_t^{d,a})_t\in\mathcal A}} \mathbb E\left[X_T + \sum_{i=1}^dq_T^iS_T^i -\frac{1}{2}\gamma \int_0^T q_t^\intercal \Sigma q_t dt\right].
\end{align*}
We call Model B our model with this second objective function.\\

\subsection{The Hamilton-Jacobi-Bellman and Hamilton-Jacobi equations}
\vspace{3mm}
Let $\{e^i\}_{i=1}^{d}$ be the canonical basis of $\Rr^d$. The Hamilton-Jacobi-Bellman equation associated with Model~A~is
\begin{eqnarray}
    0 &=&\partial_t u(t,x,q,S) + \frac 12 \sum_{i,j=1}^d \rho^{i,j} \sigma^i \sigma^j \partial^2_{S^i S^j} u(t,x,q,S) \nonumber \\
    &&+ \sum_{i=1}^d \mathds{1}_{\{q^i + z^i \le Q^i\}} \sup_{\delta^{i,b}} \Lambda^{i,b}(\delta^{i,b}) \left( u(t,x-z^iS^i + z^i \delta^{i,b}, q+z^i e^i, S) - u(t,x,q,S) \right) \nonumber \\
    &&+ \sum_{i=1}^d \mathds{1}_{\{q^i - z^i \ge- Q^i\}} \sup_{\delta^{i,a}} \Lambda^{i,a}(\delta^{i,a}) \left( u(t,x+z^iS^i + z^i \delta^{i,a}, q-z^i e^i, S) - u(t,x,q,S) \right), \label{HJB_A}
\end{eqnarray}
for all $(t,x,q,S) \in [0,T) \times \mathbb R \times \prod_{i=1}^d \left(z^i \mathbb Z \cap [-Q^i, Q^i]\right) \times \mathbb R^d$,\footnote{Given a positive number $z \in \mathbb R_+^*$, $z\mathbb Z$ denotes the set of multiples of $z$, i.e. $z\mathbb Z = \{\ldots, -2z,-z,0,z,2z,\ldots \}.$} with terminal condition $$u(T,x,q,S) = - \exp \left( - \gamma \left( x + \sum_{i=1}^d q^iS^i  \right) \right) \quad \forall (x,q,S) \in \mathbb R \times \prod_{i=1}^d \left(z^i \mathbb Z \cap  [-Q^i, Q^i]\right) \times \mathbb R^d.$$

The Hamilton-Jacobi-Bellman equation associated with Model B is
\begin{eqnarray}
    0 &=&\partial_t v(t,x,q,S) - \frac 12 \gamma q^\intercal \Sigma q + \frac 12 \sum_{i,j=1}^d \rho^{i,j} \sigma^i \sigma^j \partial^2_{S^i S^j} v(t,x,q,S) \nonumber \\
    &&+ \sum_{i=1}^d \mathds{1}_{\{q^i + z^i \le Q^i\}} \sup_{\delta^{i,b}} \Lambda^{i,b}(\delta^{i,b}) \left( v(t,x-z^iS^i + z^i \delta^{i,b}, q+z^i e^i, S) - v(t,x,q,S) \right) \nonumber \\
    &&+ \sum_{i=1}^d \mathds{1}_{\{q^i - z^i \ge - Q^i\}} \sup_{\delta^{i,a}} \Lambda^{i,a}(\delta^{i,a}) \left( v(t,x+z^iS^i + z^i \delta^{i,a}, q-z^i e^i, S) - v(t,x,q,S) \right), \label{HJB_B}
\end{eqnarray}
for all $(t,x,q,S) \in [0,T) \times \mathbb R \times \prod_{i=1}^d \left(z^i \mathbb Z \cap [-Q^i, Q^i]\right) \times \mathbb R^d$ with terminal condition $$v(T,x,q,S) =  x + \sum_{i=1}^d q^iS^i \quad \forall (x,q,S) \in \mathbb R \times \prod_{i=1}^d \left(z^i \mathbb Z \cap [-Q^i, Q^i]\right) \times \mathbb R^d.$$

For each $i \in \{1,\ldots, d\}$ and $\xi \ge 0,$ let us define two Hamiltonian functions\footnote{It is noteworthy that our definition of $H_\xi^{i,b}$ and $H_\xi^{i,a}$
differs from that of \cite{gueant2017optimal} (by a factor $z^i$). The
alternative definition we use in this paper is also that of \cite{bergault2019size} for $\xi=0$.}
$H_\xi^{i,b}$ and $H_\xi^{i,a}$ by
\begin{equation} \label{eq:Hb}
H^{i,b}_\xi(p)=
\begin{cases}
	\underset{\delta}{\sup}\frac{\Lambda^{i,b}(\delta)}{\xi z^i}(1-\exp(-\xi z^i(\delta-p)))&\mbox{if } \xi>0,\\
	\underset{\delta}{\sup}\Lambda^{i,b}(\delta)(\delta-p)&\mbox{if }\xi=0,
\end{cases}
\end{equation}
and
\begin{equation} \label{eq:Ha}
H^{i,a}_\xi(p)=
\begin{cases}
	\underset{\delta}{\sup}\frac{\Lambda^{i,a}(\delta)}{\xi z^i}(1-\exp(-\xi z^i(\delta-p)))&\mbox{if } \xi>0,\\
	\underset{\delta}{\sup}\Lambda^{i,a}(\delta)(\delta-p)&\mbox{if }\xi=0.
\end{cases}
\end{equation}

Using the ansatz introduced in \cite{gueant2017optimal} for the two functions $u: [0,T] \times \mathbb R \times \prod_{i=1}^d \left(z^i \mathbb Z \cap [-Q^i, Q^i]\right) \times \mathbb R^d \rightarrow \mathbb R$ and $v: [0,T] \times \mathbb R \times \prod_{i=1}^d \left(z^i \mathbb Z \cap [-Q^i, Q^i]\right) \times \mathbb R^d \rightarrow \mathbb R$, i.e.
    $$u(t,x,q,S) = - \exp \left(-\gamma  \left( x + \sum_{i=1}^d q^iS^i + \theta(t,q) \right) \right) \text{ and } v(t,x,q,S) = x + \sum_{i=1}^d q^iS^i + \theta(t,q),$$
we see that solving the Hamilton-Jacobi-Bellman equations \eqref{HJB_A} and \eqref{HJB_B} boils down to finding the solution $\theta : [0,T] \times \prod_{i=1}^d \left(z^i \mathbb Z \cap [-Q^i, Q^i]\right) \to \mathbb R$ of the following Hamilton-Jacobi equation with $\xi = \gamma$ in the case of Model A and $\xi =0$ in the case of Model B:
\begin{eqnarray}
    \label{sec2:thetagen}
    0 &=& \partial_t \theta(t,q) - \frac 12 \gamma q^\intercal \Sigma q\\
    &&+ \sum_{i=1}^d \mathds{1}_{\{q^i + z^i \le Q^i\}} z^i H^{i,b}_\xi \left( \frac{\theta(t,q) - \theta(t,q+z^ie^i)}{z^i} \right) + \sum_{i=1}^d \mathds{1}_{\{q^i - z^i \ge - Q^i\}} z^i H^{i,a}_\xi \left( \frac{\theta(t,q) - \theta(t,q-z^ie^i)}{z^i} \right).  \nonumber
\end{eqnarray}
In both cases, the terminal condition simply boils down to
\begin{align}
    \label{sec2:thetagenCT}
    \theta(T,q) = 0.
\end{align}

\subsection{Existing theoretical results}
\vspace{3mm}
\noindent
From \cite[Theorem 5.1]{gueant2017optimal}, for a given $\xi \ge 0,$ there exists a unique $\theta:[0,T]\times \prod_{i=1}^d \left(z^i \mathbb Z \cap [-Q^i, Q^i]\right)\to \Rr$, $C^1$ in time, solution of Eq.~\eqref{sec2:thetagen} with terminal condition \eqref{sec2:thetagenCT}.
Moreover (see \cite[Theorems 5.2 and 5.3]{gueant2017optimal}), a classical verification argument enables to go from $\theta$ to optimal controls for both Model A and Model B. The optimal quotes as functions of $\theta$ are recalled in the following theorems (for details, see \cite{gueant2017optimal}).\\

In the case of Model A, the result is the following:
\begin{theorem}
\label{optquotes}
	Let us consider the solution $\theta$ of Eq.~\eqref{sec2:thetagen} with terminal condition \eqref{sec2:thetagenCT} for $\xi=\gamma$.\\
	
	Then, for $i \in \{1, \ldots, d\},$ the optimal bid and ask quotes $S^{i,b}_t = S^i_t - \delta^{i,b*}_t$ and $S^{i,a}_t = S^i_t + \delta^{i,a*}_t$ in Model A are characterized by\begin{align}
	\begin{split}
	\label{sec2:deltaoptimalA}
	\delta^{i,b*}_t &= \tilde{\delta}^{i,b*}_\gamma\left(\frac{\theta(t,q_{t-}) - \theta(t,q_{t-}+z^i e^i)}{z^i}\right) \quad \text{for } q_{t-}+z^i e^i \in \prod_{j=1}^d \left(z^j \mathbb Z \cap [-Q^j, Q^j]\right),\\
	\delta^{i,a*}_t &= \tilde{\delta}^{i,a*}_\gamma\left(\frac{\theta(t,q_{t-}) - \theta(t,q_{t-}-z^i e^i)}{z^i}\right) \quad \text{for } q_{t-}-z^i e^i \in \prod_{j=1}^d \left(z^j \mathbb Z \cap [-Q^j, Q^j]\right),
	\end{split}
	\end{align}where the functions $\tilde{\delta}^{i,b*}_\gamma(\cdot)$ and $\tilde{\delta}_\gamma^{i,a*}(\cdot)$ are defined by
	\begin{align*}
	\tilde{\delta}^{i,b*}_\gamma(p) &= {\Lambda^{i,b}}^{-1}\left(\gamma z^i H^{i,b}_{\gamma}(p) - {H_{\gamma}^{i,b}}'(p)\right),\\
	\tilde{\delta}^{i,a*}_\gamma(p) &= {\Lambda^{i,a}}^{-1}\left(\gamma z^i H^{i,a}_{\gamma}(p) - {H_{\gamma}^{i,a}}'(p)\right),
	\end{align*}
	where for all $i \in \{1, \ldots, d\},$ ${H_{\gamma}^{i,b}}'$ and ${H_{\gamma}^{i,a}}'$ denote the first derivative of ${H_{\gamma}^{i,b}}$ and ${H_{\gamma}^{i,a}}$, respectively.
\end{theorem}
For Model B, the result is the following:
\begin{theorem}
\label{optquotesbis}
	Let us consider the solution $\theta$ of Eq.~\eqref{sec2:thetagen} with terminal condition \eqref{sec2:thetagenCT} for $\xi=0$.\\
	
	Then, for $i \in \{1, \ldots, d\},$ the optimal bid and ask quotes $S^{i,b}_t = S^i_t - \delta^{i,b*}_t$ and $S^{i,a}_t = S^i_t + \delta^{i,a*}_t$ in  Model B are characterized by\begin{align}
		\begin{split}
	\label{sec2:deltaoptimalB}
	\delta^{i,b*}_t &= \tilde{\delta}^{i,b*}_0\left(\frac{\theta(t,q_{t-}) - \theta(t,q_{t-}+z^i e^i)}{z^i}\right) \quad \text{for } q_{t-}+z^i e^i \in \prod_{j=1}^d \left(z^j \mathbb Z \cap [-Q^j, Q^j]\right), \\
	 \delta^{i,a*}_t &= \tilde{\delta}^{i,a*}_0\left(\frac{\theta(t,q_{t-}) - \theta(t,q_{t-}-z^i e^i)}{z^i}\right) \quad \text{for } q_{t-}-z^i e^i \in \prod_{j=1}^d \left(z^j \mathbb Z \cap [-Q^j, Q^j]\right),
	 \end{split}
	\end{align}where the functions $\tilde{\delta}^{i,b*}_0(\cdot)$ and $\tilde{\delta}_0^{i,a*}(\cdot)$ are defined by
	$$\tilde{\delta}^{i,b*}_0(p) = {\Lambda^{i,b}}^{-1}\left(- {H_{0}^{i,b}}'(p)\right) \text{ and } \tilde{\delta}^{i,a*}_0(p) = {\Lambda^{i,a}}^{-1}\left(- {H_{0}^{i,a}}'(p)\right)$$
where for all $i \in \{1, \ldots, d\},$ ${H_{0}^{i,b}}'$ and ${H_{0}^{i,a}}'$ denote the first derivative of ${H_{0}^{i,b}}$ and ${H_{0}^{i,a}}$, respectively.
\end{theorem}

\noindent
In the following two sections, we propose new methods to find approximations of the solution to the system of ordinary differential equations (ODEs) \eqref{sec2:thetagen} with terminal condition \eqref{sec2:thetagenCT}. Eqs. \eqref{sec2:deltaoptimalA} and \eqref{sec2:deltaoptimalB} can then serve to go from approximations of $\theta$ (hereafter called -- slightly abusively -- the value function) to approximations of the optimal quotes. The resulting quotes correspond to what the reinforcement learning community calls  the greedy quoting strategy associated with the proxy of the value function.\footnote{The true optimal quotes correspond to the greedy strategy with respect to the value function $u$ (in Model A) or $v$ (in Model~B) deduced from the true $\theta$.}

\section{A quadratic approximation of the value function and its applications} \label{sec:generaleq}

\subsection{Introduction}
\vspace{3mm}
In the field of (stochastic) optimal control, finding value functions and optimal controls in closed form is the exception rather than the rule. One important exception goes with the class of Linear-Quadratic (LQ) and Linear-Quadratic-Gaussian (LQG) problems. Of course, the above market making problem does not belong to this class of control problems, for instance because the control of point processes is nonlinear by nature. Nevertheless, we see that price risk appears in both Model A and Model B through the quadratic term $\frac 12 \gamma q^\intercal \Sigma q$ in the Hamilton-Jacobi equation \eqref{sec2:thetagen}. The main idea of this paper consists in replacing the Hamiltonian functions associated with our market making problem by quadratic functions that approximate them. The interest of quadratic Hamiltonian functions lies in that the resulting Hamilton-Jacobi equations can be solved in closed-form using the same tools as for LQ/LQG problems, i.e. Riccati equations.\\

At first sight, approximating the Hamiltonian functions involved in Eq. \eqref{sec2:thetagen} by quadratic functions seems inappropriate. For all $i \in \{1, \ldots, d\}$, the functions $H_\xi^{i,b}$ and $H_\xi^{i,a}$ are indeed positive and decreasing and approximating them with U-shaped functions can only be valid locally. However, one has to bear in mind that our goal is to approximate the solution of the Hamilton-Jacobi equations and not the Hamiltonian functions. This remark is particularly important because the Hamiltonian terms involved in the Hamilton-Jacobi equations are (up to the indicator functions that we shall discard in what follows by considering the limit case where $\forall i \in \{1, \ldots, d\}, Q^i = +\infty$) of the form $$H^{i,b}_\xi \left( \frac{\theta(t,q) - \theta(t,q+z^ie^i)}{z^i} \right) + H^{i,a}_\xi \left( \frac{\theta(t,q) - \theta(t,q-z^ie^i)}{z^i} \right),$$

Assuming that $\frac{\theta(t,q) - \theta(t,q+z^ie^i)}{z^i} \simeq - \frac{\theta(t,q) - \theta(t,q-z^ie^i)}{z^i}$, we clearly see that, with respect to asset $i$, the function we need to approximate is $p \mapsto H^{i,b}_\xi(p) + H^{i,a}_\xi(-p)$ rather than $H^{i,b}_\xi$ and $H^{i,a}_\xi$ themselves, and it is natural to approximate the former function with a U-shaped one!\\

Let us replace for all $i \in \{1, \ldots, d\}$ the Hamiltonian functions $H_\xi^{i,b}$ and $H_ \xi^{i,a}$ by the quadratic functions\footnote{We omit the subscript $\xi$ in the definition of $\check{H}^{i,b}$ and $\check{H}^{i,a}$. In particular, although the subscript $\xi$ is not written, the coefficients  $\alpha^{i,b}_0$, $\alpha^{i,b}_1$, $\alpha^{i,b}_2$, $\alpha^{i,a}_0$, $\alpha^{i,a}_1$, and $\alpha^{i,a}_2$ do depend on $\xi$.}
$$\check{H}^{i,b}: p \mapsto  \alpha^{i,b}_0 + \alpha^{i,b}_1 p + \frac 12 \alpha^{i,b}_2 p^2 \quad \textrm{and} \quad \check{H}^{i,a}: p \mapsto  \alpha^{i,a}_0 + \alpha^{i,a}_1 p + \frac 12 \alpha^{i,a}_2 p^2.$$

\begin{remark}
\label{naturalchoice}
A natural choice for the functions $(\check{H}^{i,b})_{i \in \{1, \ldots, d\}}$ and $(\check{H}^{i,a})_{i \in \{1, \ldots, d\}}$ derives from Taylor expansions around $p=0$. In that case,
$$\forall i \in \{1, \ldots, d\},  \forall j \in \{0,1,2\},\quad  \alpha^{i,b}_j = {H^{i,b}_\xi}^{(j)}(0) \quad \textrm{and} \quad \alpha^{i,a}_j = {H^{i,a}_\xi}^{(j)}(0).$$
\end{remark}

We denote by $\check \theta$ the approximation of $\theta$ associated with the functions $(\check{H}^{i,b})_{i \in \{1, \ldots, d\}}$ and $(\check{H}^{i,a})_{i \in \{1, \ldots, d\}}$, i.e. if we consider the limit case where $\forall i \in \{1, \ldots, d\}, Q^i = +\infty$, $\check \theta$ verifies
\begin{eqnarray}
0 &=& \partial_t \check{\theta}(t,q) - \frac{1}{2} \gamma q^\intercal \Sigma q
	+ \sum_{i=1}^d z^i\left( \alpha^{i,b}_0 + \alpha^{i,a}_0 \right) \nonumber \\
  &&+ \sum_{i=1}^d \left( \alpha^{i,b}_1 \left(\check{\theta}(t,q) - \check{\theta}(t, q + z^ie^i)\right)
                     + \alpha^{i,a}_1 \left(\check{\theta}(t,q) - \check{\theta}(t, q - z^ie^i)\right) \right) \nonumber \\
  &&+ \frac{1}{2}\sum_{i=1}^d \frac 1{z^i}\left( \alpha^{i,b}_2 \left(\check{\theta}(t,q) - \check{\theta}(t,q + z^ie^i)\right)^2
                                + \alpha^{i,a}_2 \left(\check{\theta}(t,q) - \check{\theta}(t,q - z^ie^i)\right)^2 \right) \label{sec3:thetagenapprox}
\end{eqnarray}
and of course we consider the terminal condition
\begin{equation}
	\label{sec3:thetagenapproxCT}
	\check{\theta}(T,q) = 0.
\end{equation}

\subsection{An approximation of the value function in closed form}
\vspace{3mm}
Eq. \eqref{sec3:thetagenapprox} with terminal condition \eqref{sec3:thetagenapproxCT} can be solved in closed form. To prove this point, we start with the following proposition:

\begin{proposition}
\label{sec3:ansatz}
Let us introduce for $i \in \{1, \ldots, d\}, j \in \{0, 1, 2\}, k \in \mathbb{N}$,
$$\Delta_{j,k}^{i,b} = \alpha^{i,b}_j (z^i)^k \quad \text{and} \quad \Delta_{j,k}^{i,a} = \alpha^{i,a}_j (z^i)^k,$$
$$V^b_{j,k} = \left(\Delta_{j,k}^{1,b},\ldots,\Delta_{j,k}^{d,b}\right)^\intercal \quad \text{and} \quad V^a_{j,k} = \left(\Delta_{j,k}^{1,a},\ldots,\Delta_{j,k}^{d,a}\right)^\intercal,$$
\begin{equation*}
D^b_{j,k} = \textup{diag}(\Delta_{j,k}^{1,b}, \ldots, \Delta_{j,k}^{d,b}) \quad \text{and} \quad D^a_{j,k} = \textup{diag}(\Delta_{j,k}^{1,a}, \ldots, \Delta_{j,k}^{d,a}).
\end{equation*}

Let us consider three differentiable functions $A: [0,T] \to S_d^{+}$, $B:
[0,T] \to \mathbb{R}^d$, and $C: [0,T] \to \mathbb R$ solutions of the system of
ordinary differential equations\footnote{$S_d^{+}$ (resp. $S_d^++$) stands throughout this paper 
the set of positive semi-definite (resp. definite) symmetric $d$-by-$d$
matrices.}
\begin{equation}\label{system:ABC}
\begin{cases}
\displaystyle
{A'}(t) = & 2 A(t)\left(D^b_{2,1} + D^a_{2,1}\right) A(t) - \frac{1}{2} \gamma \Sigma \\
{B'}(t) = &  2 A(t) \left(V^b_{1,1} - V^a_{1,1}\right) + 2A(t)\left(D^b_{2,2} - D^a_{2,2}\right) \mathcal{D}(A(t)) + 2 A(t)\left(D^b_{2,1} + D^a_{2,1}\right) B(t) \\
{C'}(t) = & \textrm{Tr}\left(D^b_{0,1} + D^a_{0,1}\right) +\textrm{Tr}\left(\left(D^b_{1,2} + D^a_{1,2}\right) A(t)\right) + \left(V^b_{1,1} - V^a_{1,1}\right)^\intercal B(t) \\
          & + \frac 12 \mathcal{D}(A(t))^\intercal \left(D^b_{2,3} + D^a_{2,3}\right) \mathcal{D}(A(t)) + \frac 12 B(t)^\intercal \left(D^b_{2,1} + D^a_{2,1}\right) B(t) + B(t)^\intercal \left(D^b_{2,2} - D^a_{2,2}\right) \mathcal{D}(A(t)),
\end{cases}
\end{equation}
with terminal conditions
\begin{equation}
\label{system:ABCCT}
A(T) = 0, B(T) = 0, \textrm{ and } C(T) = 0,
\end{equation} where $\mathcal{D}$ is the linear operator mapping a matrix onto the vector of its diagonal coefficients.\\

Then $\check{\theta}: (t,q) \in [0,T] \times \prod_{i=1}^d z^i \mathbb Z \mapsto -q^\intercal A(t) q - q^\intercal B(t) -  C(t)$ is solution of Eq. \eqref{sec3:thetagenapprox} with terminal condition~\eqref{sec3:thetagenapproxCT}.
\end{proposition}

\begin{proof}
We have
\vspace{-3mm}
\begingroup
\allowdisplaybreaks
\begin{eqnarray*}
&& \partial_t \check{\theta}(t,q) - \frac{1}{2} \gamma q^\intercal \Sigma q
\vspace{-3mm}
	+ \sum_{i=1}^d z^i\left( \alpha^{i,b}_0 + \alpha^{i,a}_0 \right)\\*
  &&+ \sum_{i=1}^d \left( \alpha^{i,b}_1 \left(\check{\theta}(t,q) - \check{\theta}(t, q + z^ie^i)\right)
                     + \alpha^{i,a}_1 \left(\check{\theta}(t,q) - \check{\theta}(t, q - z^ie^i)\right) \right) \\*
  &&+ \frac{1}{2}\sum_{i=1}^d \frac 1{z^i}\left( \alpha^{i,b}_2 \left(\check{\theta}(t,q) - \check{\theta}(t,q + z^ie^i)\right)^2
                                + \alpha^{i,a}_2 \left(\check{\theta}(t,q) - \check{\theta}(t,q - z^ie^i)\right)^2 \right)\\
 &=& -q^\intercal A'(t) q - q^\intercal B'(t) - C'(t) - \frac{1}{2} \gamma q^\intercal \Sigma q + \sum_{i=1}^d z^i( \alpha^{i,b}_0 + \alpha^{i,a}_0 )\\
\vspace{-3mm} &&+ \sum_{i=1}^d \alpha^{i,b}_1 \left(2z^i q^\intercal A(t) e^i + (z^i)^2 {e^i}^\intercal A(t) e^i + z^i {e^i}^\intercal B(t) \right)\\*
 &&+ \sum_{i=1}^d \alpha^{i,a}_1 \left(-2z^i q^\intercal A(t) e^i + (z^i)^2 {e^i}^\intercal A(t) e^i - z^i {e^i}^\intercal B(t)\right) \\*
  &&+ \frac{1}{2}\sum_{i=1}^d \frac 1{z^i} \alpha^{i,b}_2 \left(2z^i q^\intercal A(t) e^i + (z^i)^2 {e^i}^\intercal A(t) e^i + z^i {e^i}^\intercal B(t)\right)^2\\*
  &&+ \frac{1}{2}\sum_{i=1}^d \frac 1{z^i} \alpha^{i,a}_2 \left(-2z^i q^\intercal A(t) e^i + (z^i)^2 {e^i}^\intercal A(t) e^i - z^i {e^i}^\intercal B(t)\right)^2\\
  &=& -q^\intercal A'(t) q - q^\intercal B'(t) - C'(t) - \frac{1}{2} \gamma q^\intercal \Sigma q + \textrm{Tr}\left(D^b_{0,1} + D^a_{0,1}\right)\\*
 &&+ 2 q^\intercal A(t) \left(V^b_{1,1} - V^a_{1,1}\right) + \textrm{Tr}\left(\left(D^b_{1,2} + D^a_{1,2}\right) A(t)\right) + \left(V^b_{1,1} - V^a_{1,1}\right)^\intercal B(t) \\*
  &&+ 2 q^\intercal A(t) \left(D^b_{2,1} + D^a_{2,1}\right) A(t) q + \frac 12 \mathcal{D}(A(t))^\intercal \left(D^b_{2,3} + D^a_{2,3}\right) \mathcal{D}(A(t)) + \frac 12 B(t)^\intercal \left(D^b_{2,1} + D^a_{2,1}\right) B(t)\\*
  &&+2 q^\intercal A(t) \left(D^b_{2,2} - D^a_{2,2}\right) \mathcal{D}(A(t)) + 2 q^\intercal A(t) \left(D^b_{2,1} + D^a_{2,1}\right)B(t) + B(t)^\intercal \left(D^b_{2,2} - D^a_{2,2}\right) \mathcal{D}(A(t))\\
  &=&0,
\end{eqnarray*}
\endgroup
where the last equality comes from the definitions of $(A,B,C)$ and the identification of the terms of degree $0$, $1$, and $2$ in $q$.\\

As the terminal conditions are satisfied, the result is proved.
\end{proof}

\begin{proposition}\label{prop:ode_solution}
Assume $\alpha^{i, b}_2 + \alpha^{i, a}_2 > 0$ for all $i \in \{ 1, \ldots, d
\}$. The system of ODEs \eqref{system:ABC} with terminal conditions
\eqref{system:ABCCT} admits the unique solution
\begin{align}
A(t) & = \frac{1}{2} D^{-\frac 12}_+ \widehat{A}
    \left( e^{ \widehat{A} (T - t)} - e^{-\widehat{A} (T - t)} \right)
    \left( e^{ \widehat{A} (T - t)} + e^{-\widehat{A} (T - t)} \right)^{-1} D_+^{-\frac 12},
         \label{eq:solA} \\
B(t) & = -2 e^{-2 \int_t^T A(u) D_+ \,du} \int_t^T
            e^{ 2 \int_s^T A(u) D_+ \,du} A(s)
            \left( V_- + D_- \mathcal D(A(s)) \right) ds,
            \label{eq:solB} \\
C(t) & = -\textrm{Tr}\left(D^b_{0,1} + D^a_{0,1}\right) (T - t)
       -  \textrm{Tr}\left(\left(D^b_{1,2} + D^a_{1,2}\right) \int_t^T A(s) ds \right)
       -  V_-^\intercal  \int_t^T B(s) ds \nonumber \\
     & - \frac 12 \int_t^T \mathcal{D}(A(s))^\intercal \left(D^b_{2,3}
         + D^a_{2,3}\right) \mathcal{D}(A(s)) ds
       - \frac 12 \int_t^T B(s)^\intercal D_+ B(s) ds  \nonumber \\
     & - \int_t^T B(s)^\intercal D_- \mathcal{D}(A(s)) ds.\label{eq:solC}
\end{align}
where
\[ D_+ = D^b_{2,1} + D^a_{2,1}, \quad
   D_- = D_{2,2}^b - D_{2,2}^a, \quad
   V_- = V_{1,1}^b - V_{1,1}^a, \quad \text{and} \quad
   \widehat{A}
   = \sqrt{\gamma} \left( D_+^{\frac 12} \Sigma D_+^{\frac 12} \right)^{\frac 12}. \]
\end{proposition}

\begin{proof}
The system of ODEs \eqref{system:ABC} being triangular -- though not linear -- we tackle the equations one by one, in order.\\

\paragraph{\textbf{Solution for $A$}}

First, we observe that $D_+ = \textup{diag}((\alpha^{1, b}_2 +
\alpha^{1, a}_2) z^1, \ldots, (\alpha^{d, b}_2 + \alpha^{d, a}_2) z^d)$ is a
positive diagonal matrix. Therefore $D_+^{\frac 12}$ is well defined. Then, since $D_+^{\frac 12} \Sigma D_+^{\frac 12} \in S_d^{+}$, $\widehat{A}$ is well defined and in $S_d^{+}$.\\

Now, by introducing the
change of variables
\[ \mathbf{a}(t) = 2 D_+^{\frac 12} A(t) D_+^{\frac 12}, \]
the terminal value problem for $A$ in \eqref{system:ABC} becomes
\begin{equation}\label{eq:bolda}
\begin{cases}
&\mathbf{a}'(t) = \mathbf{a}(t)^2 - \widehat{A}^2 \\
&\mathbf{a}(T) = 0.
\end{cases}
\end{equation}
To solve  \eqref{eq:bolda} let us introduce the function $z$ defined by
\[ z(t) = e^{\widehat{A}(T-t)} + e^{-\widehat{A}(T-t)}, \]
that is a $C^2([0,T], S_d^{++})$ function verifying $z''(t) = \widehat{A}^2 z(t)$ and $z'(T) = 0$.\\

We have
$$\frac{d}{dt}\left(-z'(t) z(t)^{-1}\right) = -z''(t) z(t)^{-1} + z'(t) z(t)^{-1}z'(t) z(t)^{-1} = \left(z'(t) z(t)^{-1}\right)^2 -\widehat{A}^2$$ and $-z'(T) z(T)^{-1} = 0$. Therefore, by Cauchy-Lipschitz theorem, we have $\mathbf{a} = -z'z^{-1}$.\\

Wrapping up, we obtain
\begin{align*}
A(t)
& = \frac{1}{2} D_+^{-\frac 12} \mathbf{a}(t) D_+^{-\frac 12} \\
& = -\frac{1}{2} D_+^{-\frac 12} z'(t) z(t)^{-1} D_+^{-\frac 12} \\
& = \frac{1}{2} D_+^{-\frac 12} \widehat{A}
    \left( e^{ \widehat{A} (T - t)} - e^{-\widehat{A} (T - t)} \right)
    \left( e^{ \widehat{A} (T - t)} + e^{-\widehat{A} (T - t)} \right)^{-1}
    D_+^{-\frac 12}.
\end{align*}

\paragraph{\textbf{Solution for $B$}}
Let us notice that, by definition of the exponential of a matrix, for all $s,t \in [0,T]$, the matrices $\widehat{A}$, $\left(
e^{ \widehat{A} (T - s)} - e^{-\widehat{A} (T - s)} \right)$, $\left(
e^{ \widehat{A} (T - s)} + e^{-\widehat{A} (T - s)} \right)^{-1}$, $\left( e^{
\widehat{A} (T - t)} - e^{-\widehat{A} (T - t)} \right)$, and $\left( e^{
\widehat{A} (T - t)} + e^{-\widehat{A} (T - t)} \right)^{-1}$ commute. Therefore
\begin{eqnarray*}
&&A(s) D_+ A(t) D_+\\
&=& \frac{1}{4} D^{-\frac 12}_+ \widehat{A}
    \left( e^{ \widehat{A} (T - s)} - e^{-\widehat{A} (T - s)} \right)
    \left( e^{ \widehat{A} (T - s)} + e^{-\widehat{A} (T - s)} \right)^{-1}\\
&&\times\widehat{A}
    \left( e^{ \widehat{A} (T - t)} - e^{-\widehat{A} (T - t)} \right)
    \left( e^{ \widehat{A} (T - t)} + e^{-\widehat{A} (T - t)} \right)^{-1}D_+^{\frac 12}\\
    &=& \frac{1}{4} D^{-\frac 12}_+ \widehat{A}
    \left( e^{ \widehat{A} (T - t)} - e^{-\widehat{A} (T - t)} \right)
    \left( e^{ \widehat{A} (T - t)} + e^{-\widehat{A} (T - t)} \right)^{-1}\\
&&\times\widehat{A}
    \left( e^{\widehat{A} (T - s)} - e^{-\widehat{A} (T - s)} \right)
    \left( e^{ \widehat{A} (T - s)} + e^{-\widehat{A} (T - s)} \right)^{-1}D_+^{\frac 12}\\
&=& A(t) D_+ A(s) D_+.
\end{eqnarray*}

Therefore, we can apply the method of Variation of Parameters to the linear ODE characterizing $B$ to obtain
$$B(t) = -2 e^{-2 \int_t^T A(u) D_+ \,du} \int_t^T
            e^{ 2 \int_s^T A(u) D_+ \,du} A(s)
            \left( V_- + D_- \mathcal D(A(s)) \right) ds.$$

\paragraph{\textbf{Solution for $C$}}
We simply integrate the ODE characterizing $C$ to obtain $\eqref{eq:solC}$.

\end{proof}

From Eqs. \eqref{eq:solA}, \eqref{eq:solB}, and \eqref{eq:solC}, we can deduce the asymptotic behaviour of $(A,B,C)$ when $T$ goes to infinity.\\

\begin{proposition}\label{prop:ode_asymp_solution}
Let $(A,B,C)$ be the solution of the system of ODEs \eqref{system:ABC} with terminal conditions \eqref{system:ABCCT}.\\

Then,
\begin{align*}
A(0) & \stackrel{T\to+\infty}{\longrightarrow} \frac{1}{2}\sqrt{\gamma}\Gamma, \\
B(0) & \stackrel{T\to+\infty}{\longrightarrow} - D_+^{-\frac 12}\widehat{A}\widehat{A}^+ D_+^{-\frac 12} \left(V_-
     + \frac{1}{2}\sqrt{\gamma} D_- \mathcal D(\Gamma)\right),\\
\frac{C(0)}T & \stackrel{T\to+\infty}{\longrightarrow} -\textrm{Tr}\left(D^b_{0,1} + D^a_{0,1}\right) -\frac{1}{2}\sqrt{\gamma} \textrm{Tr}\left(\left(D^b_{1,2} + D^a_{1,2}\right)\Gamma\right) + V_-^\intercal D_+^{-\frac 12}\widehat{A}\widehat{A}^+ D_+^{-\frac 12} \left(V_- + \frac{1}{2}\sqrt{\gamma} D_- \mathcal D(\Gamma)\right) \\
          &\qquad - \frac 18 \gamma \mathcal{D}(\Gamma)^\intercal \left(D^b_{2,3}\! +\! D^a_{2,3}\right) \mathcal{D}(\Gamma) - \frac 12 \left(V_- \!
     +\! \frac{1}{2}\sqrt{\gamma} D_- \mathcal D(\Gamma)\right)^\intercal\! D_+^{-\frac 12}\widehat{A}\widehat{A}^+ D_+^{-\frac 12} \left(V_-
     \!+\! \frac{1}{2}\sqrt{\gamma} D_- \mathcal D(\Gamma)\right)\\
     &\qquad + \frac{1}{2}\sqrt{\gamma} \left(V_-
     + \frac{1}{2}\sqrt{\gamma} D_- \mathcal D(\Gamma)\right)^\intercal D_+^{-\frac 12}\widehat{A}\widehat{A}^+ D_+^{-\frac 12}  D_- \mathcal{D}(\Gamma),
\end{align*}
where $\Gamma = D_+^{-\frac 1 2} \left( D_+^{\frac 12} \Sigma D_+^{\frac 12} \right)^{\frac 12}
D_+^{-\frac 1 2}$ and $\widehat{A}^+$ is the Moore-Penrose generalized inverse of
$\widehat{A}$.
\end{proposition}

\begin{proof}
This proof is divided into three parts corresponding to the derivation of the asymptotic
expression for $A$, $B$, and $C$, respectively.\\

\paragraph{\textbf{Asymptotics for $A$}}
Let us recall first that $\widehat{A}
   = \sqrt{\gamma} \left( D_+^{\frac 12} \Sigma D_+^{\frac 12} \right)^{\frac 12} \in S_d^+$. Therefore, there exists an orthogonal matrix $P$ and there exists a diagonal matrix with nonnegative entries $\textup{diag}(\lambda_1, \ldots, \lambda_d)$ such that $\widehat{A} = P \textup{diag}(\lambda_1, \ldots, \lambda_d) P^{\intercal}$. From Eq. \eqref{eq:solA} we have
\[ A(0) = \frac{1}{2} D^{-\frac 12}_+ P \textup{diag}\left(\lambda_1 \tanh(\lambda_1 T), \ldots, \lambda_d \tanh(\lambda_d T)\right) P^\intercal D_+^{-\frac 12}. \]
   As $\lambda \tanh(\lambda T) \stackrel{T \to +\infty}{\longrightarrow}
\begin{cases}
  0, & \text{if } \lambda = 0 \\
  \lambda, & \text{if } \lambda > 0
\end{cases}$, we clearly have
\[ A(0)
\stackrel{T \to +\infty}{\longrightarrow}
   \frac{1}{2} D^{-\frac 12}_+ P \textup{diag}(\lambda_1, \ldots, \lambda_d)  P^\intercal D_+^{-\frac 12} = \frac{1}{2} D^{-\frac 12}_+ \widehat{A} D_+^{-\frac 12}
=  \frac{1}{2} \sqrt{\gamma} \Gamma. \]

\paragraph{\textbf{Asymptotics for $B$}}

From Eq. \eqref{eq:solB}, we have
\begin{eqnarray*}
B(0) &=& -2 e^{-2 \int_0^T A(u) D_+ \,du} \int_0^T
            e^{ 2 \int_s^T A(u) D_+ \,du} A(s)
            \left( V_- + D_- \mathcal D(A(s)) \right) ds\\
&=& -2 e^{-2 \int_0^T \tilde A(u) D_+ \,du} \int_0^T
            e^{ 2 \int_0^s \tilde A(u) D_+ \,du} \tilde A(s)
            \left( V_- + D_- \mathcal D(\tilde A(s)) \right) ds
\end{eqnarray*}
where $\tilde A : t \mapsto \frac{1}{2} D^{-\frac 12}_+ \widehat{A}
    \left( e^{\widehat{A}t} - e^{-\widehat{A}t} \right)
    \left( e^{\widehat{A}t} + e^{-\widehat{A}t} \right)^{-1} D^{-\frac 12}_+$.\\

Using the spectral decomposition of $\widehat A$ introduced in the above paragraph, we see that
$$ 2 \tilde A(u) D_+ = D_+^{-\frac 12} P \textup{diag}\left(\lambda_1 \tanh(\lambda_1 u), \ldots, \lambda_d \tanh(\lambda_d u)\right) P^\intercal D_+^{\frac 12}$$
and therefore, after integration,
$$e^{ 2 \int_0^T \tilde A(u) D_+ \,du} = D_+^{-\frac 12} P \textup{diag}\left(\cosh(\lambda_1 T), \ldots, \cosh(\lambda_d T)\right) P^\intercal D_+^{\frac 12}$$
and
$$e^{ 2 \int_0^s \tilde A(u) D_+ \,du} \tilde A(s) = \frac 12 D_+^{-\frac 12} P \textup{diag}\left(\lambda_1\sinh(\lambda_1 s), \ldots, \lambda_d\sinh(\lambda_d s)\right) P^\intercal D_+^{-\frac 12}$$
Wrapping up, we get that $B(0)$ is equal to 
\begin{eqnarray*}
&& -\int_0^T D_+^{-\frac 12} P \textup{diag}\left(\lambda_1 \frac{\sinh(\lambda_1 s)}{\cosh(\lambda_1 T)}, \ldots, \lambda_d \frac{\sinh(\lambda_d s)}{\cosh(\lambda_d T)}\right) P^\intercal D_+^{-\frac 12} \left( V_- + D_- \mathcal D(\tilde A(s)) \right) ds.\\
&=& -\int_0^T D_+^{-\frac 12} P \textup{diag}\left(\lambda_1 \frac{\sinh(\lambda_1 s)}{\cosh(\lambda_1 T)}, \ldots, \lambda_d \frac{\sinh(\lambda_d s)}{\cosh(\lambda_d T)}\right) P^\intercal D_+^{-\frac 12} \left( V_- + \frac 12 \sqrt{\gamma} D_- \mathcal D(\Gamma) \right) ds \\
&& + \int_0^T D_+^{-\frac 12} P \textup{diag}\left(\lambda_1 \frac{\sinh(\lambda_1 s)}{\cosh(\lambda_1 T)}, \ldots, \lambda_d \frac{\sinh(\lambda_d s)}{\cosh(\lambda_d T)}\right) P^\intercal D_+^{-\frac 12} \left(\frac 12 \sqrt{\gamma} D_- \mathcal D(\Gamma) - D_- \mathcal D(\tilde A(s)) \right) ds\\
&=& D_+^{-\frac 12} P \textup{diag}\left(1 - \frac{1}{\cosh(\lambda_1 T)}, \ldots, 1- \frac{1}{\cosh(\lambda_d T)}\right) P^\intercal D_+^{-\frac 12}  \left( V_- + \frac 12 \sqrt{\gamma} D_- \mathcal D(\Gamma) \right) + J(T),\\
\end{eqnarray*}
where $$J(T) = \int_0^T D_+^{-\frac 12} P \textup{diag}\left(\lambda_1 \frac{\sinh(\lambda_1 s)}{\cosh(\lambda_1 T)}, \ldots, \lambda_d \frac{\sinh(\lambda_d s)}{\cosh(\lambda_d T)}\right) P^\intercal D_+^{-\frac 12} \left(\frac 12 \sqrt{\gamma} D_- \mathcal D(\Gamma) - D_- \mathcal D(\tilde A(s)) \right) ds.$$

Let us prove that $J(T) \stackrel{T \to +\infty}{\longrightarrow} 0$. For that
purpose, let us consider $\epsilon > 0$ and let us notice that there exists
$\tau > 0$ such that $\forall s > \tau, \|\frac 12 \sqrt{\gamma} D_+^{-1} D_-
\mathcal D(\Gamma) - D_+^{-1} D_- \mathcal D(\tilde A(s))\| \le \epsilon$, where
the norm used is the Euclidian norm on $\mathbb{R}^d$. Let us also denote by $M$
the quantity $\sup_{s \ge 0} \|\frac 12 \sqrt{\gamma} D_- \mathcal D(\Gamma) -
D_- \mathcal D(\tilde A(s))\|$.\\

Using the operator norm (still denoted by $\|\cdot\|$) associated with the
Euclidian norm on $\mathbb{R}^d$ and its well-known link with the spectral
radius, we see that for $T > \tau$,
\begin{eqnarray*}
 \!\!&&\|J(T)\|\\
 \!\!&\le&\!\!  \int_0^T \left\|D_+^{-\frac 12} P \textup{diag}\left(\lambda_1 \frac{\sinh(\lambda_1 s)}{\cosh(\lambda_1 T)}, \ldots, \lambda_d \frac{\sinh(\lambda_d s)}{\cosh(\lambda_d T)}\right) P^\intercal D_+^{\frac 12}\right\| \left\|D_+^{-1}\frac 12 \sqrt{\gamma} D_- \mathcal D(\Gamma)\! -\! D_+^{-1}D_- \mathcal D(\tilde A(s))\right\| ds\\
\!\!&\le&\!\! \int_0^T \left(\max_i \lambda_i \frac{\sinh(\lambda_i s)}{\cosh(\lambda_i T)}\right)\left\|\frac 12 \sqrt{\gamma} D_+^{-1}D_- \mathcal D(\Gamma) - D_+^{-1}D_- \mathcal D(\tilde A(s))\right\| ds\\
&\le& M \int_0^\tau \max_i \lambda_i \frac{\sinh(\lambda_i s)}{\cosh(\lambda_i T)}ds + \epsilon \int_\tau^T \max_i \lambda_i \frac{\sinh(\lambda_i s)}{\cosh(\lambda_i T)}ds.\\
\end{eqnarray*}
By defining $\overline{\lambda} = \max\{\lambda_1,
\ldots, \lambda_d\}$ and $\underline{\lambda} = \min\{\lambda_i | \forall i \in \{1, \ldots, d\}, \lambda_i > 0\}$, we have
$$\max_{i \in \{1, \ldots, d\}} \lambda_i \frac{\sinh(\lambda_i s)}{\cosh(\lambda_i T)} \le \max_{i \in \{1, \ldots, d\}} \lambda_i \frac{e^{\lambda_i s}}{e^{\lambda_i T}} =  \max_{i \in \{1, \ldots, d\}, \lambda_i > 0}  \lambda_i e^{-\lambda_i (T-s)} \le \overline{\lambda}  e^{-\underline{\lambda} (T-s)}.$$
Therefore,
\begin{eqnarray*}
&&\limsup_{T \to \infty} \|J(T)\|\\
&\le& M \limsup_{T \to \infty} \overline{\lambda} \left(e^{-\underline{\lambda} (T-\tau)} - e^{-\underline{\lambda} T}\right)
    + \epsilon \limsup_{T \to \infty} \overline{\lambda} \left(1 - e^{-\underline{\lambda} (T -\tau)}\right)\\
&\le& \epsilon
\end{eqnarray*}
which allows to conclude that $J(T) \stackrel{T \to +\infty}{\longrightarrow} 0$.\\

Now, as $P \textup{diag}\left(1 - \frac{1}{\cosh(\lambda_1 T)}, \ldots, 1-
\frac{1}{\cosh(\lambda_d T)}\right) P^\intercal$ converges toward the
orthogonal projector on $\text{Im}(\widehat A)$, which is also
given by $\widehat{A}\widehat{A}^+$, we conclude that
$$ B(0) \stackrel{T\to+\infty}{\longrightarrow} - D_+^{-\frac 12} \widehat{A}\widehat{A}^+ D_+^{-\frac 12} \left(V_-
     + \frac{1}{2}\sqrt{\gamma} D_- \mathcal D(\Gamma)\right).$$

\paragraph{\textbf{Asymptotics for $C$}}

The asymptotic behavior of $C$ is a straightforward consequence of that of $A$ and $B$.\\
\end{proof}

\subsection{From value functions to heuristics and quotes}
\vspace{3mm}
\subsubsection{Motivation for closed-form approximations}
\vspace{3mm}
An approximation in closed form of the value function can be motivated by its numerous applications. In the following, we highlight three of them.\\

First, it can serve as a heuristic evaluation function in reinforcement learning algorithms. Indeed, in problems where the time horizon is too far away to consider full exploration in time, it is often useful, when using Monte-Carlo-based reinforcement learning techniques, to proxy the value of states in a tractable way -- analogous to algorithms such as Deep Blue. The above closed-form approximations can be used for that purpose. Moreover, because the value of $C(t)$ is irrelevant for comparing two states (it vanishes when computing the difference in the value function between two points), it is sometimes possible, especially when $T$ is large, to consider the asymptotic expression $$-\frac 12 \sqrt{\gamma} q^\intercal \Gamma q + q^\intercal D_+^{-\frac 12}\widehat{A}\widehat{A}^+ D_+^{-\frac 12} \left(V_- + \frac{1}{2}\sqrt{\gamma} D_- \mathcal D(\Gamma)\right) $$ instead of $\check\theta(t,q)$.\\

Second, a closed-form approximation of the value function can be used as a starting point in iterative methods designed to compute the value function (value iteration algorithm, actor-critic algorithms, etc.). Unlike for the above use, the value of $C(t)$ matters in that case.\\

A third important application, and the one that initially motivated our paper, is for computing policies (quotes, in our case). Indeed, a policy can be deduced from an approximation of the value function by computing the greedy strategy associated with that approximation. In our market making problem, the quotes obtained in this way are not only easy to compute, but also have the advantage of being easily interpretable.\\

\subsubsection{Quotes: the general case}
\label{subapproxgen}
\vspace{3mm}
The greedy quoting strategy associated with our closed-form proxy of the value function leads to the following quotes for all $i \in \{1, \ldots, d\}$:

\begin{eqnarray*}
	\check\delta^{i,b}_t &=& \tilde{\delta}^{i,b*}_\xi\left(\frac{\check\theta(t,q_{t-}) - \check\theta(t,q_{t-}+z^i e^i)}{z^i}\right)\\
 &=& \tilde{\delta}^{i,b*}_\xi\left(2q_{t-}^\intercal A(t) e^i + z^i {e^i}^\intercal A(t) e^i + {e^i}^\intercal B(t) \right),\\
	\check\delta^{i,a}_t &=& \tilde{\delta}^{i,a*}_\xi\left(\frac{\check\theta(t,q_{t-}) - \check\theta(t,q_{t-}-z^i e^i)}{z^i}\right)\\
&=& \tilde{\delta}^{i,a*}_\xi\left(-2q_{t-}^\intercal A(t) e^i + z^i {e^i}^\intercal A(t) e^i - {e^i}^\intercal B(t) \right),
	\end{eqnarray*}
where $\tilde{\delta}^{i,b*}_\xi$ and $\tilde{\delta}^{i,a*}_\xi$ are given in Theorems \ref{optquotes} and \ref{optquotesbis} for $\xi = \gamma$ and $\xi = 0$ respectively (depending on whether one considers Model A or Model B).\\

The asymptotic regime exhibited in the above paragraphs can then serve to obtain the following simple closed-form approximations:
\begin{eqnarray}
	\breve\delta^{i,b}_t &=& \tilde{\delta}^{i,b*}_\xi\left(\sqrt\gamma q_{t-}^\intercal \Gamma e^i + \frac 12 \sqrt\gamma z^i {e^i}^\intercal \Gamma e^i - {e^i}^\intercal D_+^{-\frac 12}\widehat{A}\widehat{A}^+ D_+^{-\frac 12} \left(V_-
     + \frac{1}{2}\sqrt{\gamma} D_- \mathcal D(\Gamma) \right)\right),\label{asymptb}\\
	\breve\delta^{i,a}_t &=& \tilde{\delta}^{i,a*}_\xi\left(-\sqrt\gamma q_{t-}^\intercal \Gamma e^i + \frac 12 \sqrt\gamma z^i {e^i}^\intercal \Gamma e^i + {e^i}^\intercal D_+^{-\frac 12}\widehat{A}\widehat{A}^+ D_+^{-\frac 12} \left(V_-
     + \frac{1}{2}\sqrt{\gamma} D_- \mathcal D(\Gamma) \right)\right).\label{asympta}
	\end{eqnarray}

It is interesting to notice here that the closed-form approximation of the optimal bid and ask quotes for asset~$i$ depend on the current value of the inventory through the term $q_{t-}^\intercal \Gamma e^i$. Since $\Gamma \in S_d^+$ and the functions $\tilde{\delta}^{i,b*}_\xi$ and $\tilde{\delta}^{i,a*}_\xi$ are monotone, we have that, all else equal, the quotes for asset $i$ depend monotonically on the inventory in asset $i$ (the bid and ask prices decrease (resp. increase) when the inventory is positive (resp. negative)). The dependence on the inventory in other assets is more subtle as it is linked to the matrix $\Gamma=D_+^{-\frac 1 2}(D_+^{\frac 12}\Sigma D_+^{\frac 12})^{\frac 12} D_+^{-\frac 1 2}$ which models the complex interplay between price risk and liquidity risk. Also, as already noted in \cite{gueant2013dealing} the influence of the risk aversion parameter $\gamma$ is ambiguous and depends on the value of inventories. \\

In the case of symmetric intensities, i.e. when $\Lambda^{i,b} = \Lambda^{i,a}$  for all $i \in \{1, \ldots, d\}$, the Hamiltonian functions $H_\xi^{i,b}$ and $H_xi^{i,a}$ given in Eqs. \eqref{eq:Hb} and \eqref{eq:Ha} are identical and thus it is natural to set $\check{H}^{i,b} = \check{H}^{i,a}$ for all $i \in \{1, \ldots, d\}$. In that case, \eqref{asymptb} and \eqref{asympta} simplify into
\begin{eqnarray}
	\breve\delta^{i,b}_t &=& \tilde{\delta}^{i,b*}_\xi\left(\sqrt\gamma q_{t-}^\intercal \Gamma e^i + \frac 12 \sqrt\gamma z^i {e^i}^\intercal \Gamma e^i\right),\label{asymptbsym}\\
	\breve\delta^{i,a}_t &=& \tilde{\delta}^{i,a*}_\xi\left(-\sqrt\gamma q_{t-}^\intercal \Gamma e^i + \frac 12 \sqrt\gamma z^i {e^i}^\intercal \Gamma e^i\right).\label{asymptasym}
	\end{eqnarray}

All these approximations of the optimal quotes can be used directly or as starting points in iterative methods designed to compute the optimal quotes (policy iteration algorithm, actor-critic algorithms, etc.).\\

\subsubsection{Quotes: the case of symmetric exponential intensities}
\vspace{3mm}
Exponential intensity functions play an important role in the optimal market making literature and more generally in the algorithmic trading literature. This shape of intensity functions, initially proposed by Avellaneda and Stoikov in \cite{avellaneda2008high}, leads indeed to simplification because of the form of the associated Hamiltonian functions.\\

If we assume that the intensity functions are given, for all $i \in \{1, \ldots, d\}$, by
\[    \Lambda^{i,b} (\delta) = \Lambda^{i,a} (\delta) = A^{i} e^{-k^{i} \delta}, \quad A^i, k^i > 0,\]
then (see \cite{gueant2017optimal}) the Hamiltonian functions are given, for all $i \in \{1, \ldots, d\}$, by
\[H^{i,b}_\xi (p)
   =  H^{i,a}_\xi (p) = \frac{A^{i}}{k^{i}} C^{i}_\xi \exp(-k^{i} p), \]
 where
 \[ C_\xi^i =
 \begin{cases}
 \left(1+\frac{\xi z^i}{k^{i}}\right)^{-\left(1+\frac{k^{i}}{\xi z^i}\right)}&\mbox{if }\xi>0\\
 e^{-1}&\mbox{if }\xi=0,
 \end{cases} \]
 and the functions $\tilde{\delta}^{i,b*}_\xi$ and $\tilde{\delta}^{i,a*}_\xi$ are given, for all $i \in \{1, \ldots, d\}$, by
\[\tilde\delta^{i,b*}_\xi (p)
   =  \tilde\delta^{i,a*}_\xi (p) = \begin{cases}
 p + \frac{1}{\xi z^i}\log\left(1+\frac{\xi z^i}{k^{i}}\right)&\mbox{if }\xi>0\\
 p + \frac{1}{k^i}&\mbox{if }\xi=0.
 \end{cases} \]
Therefore, if we consider the quadratic approximation of the Hamiltonian functions based on their Taylor expansion around $p=0$ (see Remark \ref{naturalchoice}), then \eqref{asymptbsym} and \eqref{asymptasym} become

\begin{eqnarray*}
	\breve\delta^{i,b}_t &=& \begin{cases}
 \sqrt{\gamma} \left( q_{t-}^\intercal \Gamma e^i + \frac 12 z^i {e^i}^\intercal \Gamma e^i \right)
 +\frac{1}{\gamma z^i}\log\left(1+\frac{\gamma z^i}{k^{i}}\right)  &\mbox{ in Model A,}\\
 \sqrt{\gamma} \left(  q_{t-}^\intercal \Gamma e^i + \frac 12 z^i {e^i}^\intercal \Gamma e^i \right)
 +\frac{1}{k^i} &\mbox{ in  Model B.}
\end{cases}\\
	\breve\delta^{i,a}_t &=& \begin{cases}
-\sqrt{\gamma} \left( q_{t-}^\intercal \Gamma e^i - \frac 12 z^i {e^i}^\intercal \Gamma e^i \right)
+\frac{1}{\gamma z^i}\log\left(1+\frac{\gamma z^i}{k^{i}}\right)  &\mbox{ in Model A,}\\
-\sqrt{\gamma} \left( q_{t-}^\intercal \Gamma e^i - \frac 12 z^i {e^i}^\intercal \Gamma e^i \right)
+\frac{1}{k^i} &\mbox{ in Model B.}
\end{cases}
	\end{eqnarray*}
where $\Gamma=D_+^{-\frac 1 2}(D_+^{\frac 12}\Sigma D_+^{\frac 12})^{\frac 12} D_+^{-\frac 1 2}$ and $D_+=\textup{diag}(2A^1 C_\xi^1 k^1 z^1 ,\ldots, 2A^d C_\xi^d k^d z^d).$\\

It is noteworthy that these approximations of the optimal quotes are affine in the current inventory. In particular, in the case of Model A, when the number of assets is reduced to one (with unitary transaction size), they coincide with the affine closed-form approximations obtained in the paper \cite{gueant2013dealing} of Guéant, Lehalle, and Fernandez-Tapia. Their approximations, however, are obtained in a fundamentally different manner, by using spectral arguments and a continuous approximation of the initial discrete problem.\\

Another useful point of view on the above quoting strategy is by observing the resulting approximations of the optimal (half) bid-ask spread and skew. The approximations of the optimal (half) bid-ask spread and skew for asset $i$ are respectively given by
\begin{eqnarray*}
	\frac{\breve\delta^{i,a}_t + \breve\delta^{i,b}_t}{2} &=& \begin{cases}
 \frac 12 \sqrt{\gamma} z^i {e^i}^\intercal \Gamma e^i
 +\frac{1}{\gamma z^i}\log\left(1+\frac{\gamma z^i}{k^{i}}\right)  &\mbox{ in Model A,}\\
 \frac 12 \sqrt{\gamma} z^i {e^i}^\intercal \Gamma e^i
 +\frac{1}{k^i} &\mbox{ in  Model B,}
\end{cases}
\end{eqnarray*}
and
\begin{eqnarray*}
	\frac{\breve\delta^{i,a}_t - \breve\delta^{i,b}_t}{2} &=&
-\sqrt{\gamma} q_{t-}^\intercal \Gamma e^i \text{ in both Model A and Model B.}
	\end{eqnarray*}

These approximations give us a constant bid-ask spread and a skew linear in the inventory. This translates well the intuition that the skew has the role of inventory risk management, whereas the spread balances the trade-off between frequency of transactions and profit per round-trip trade (the term $\frac{1}{\gamma z^i}\log\left(1+\frac{\gamma z^i}{k^{i}}\right)$ in Model A, which reduces to $\frac{1}{k^i}$ in the case of Model B\footnote{It is noteworthy that in the case of Model B  the bid-ask spread is a nondecreasing function of the risk aversion parameter~$\gamma$.}), plus an additional risk aversion buffer (the term $\frac 12 \sqrt{\gamma} z^i {e^i}^\intercal \Gamma e^i$).\\

On the parameter sensitivity analysis, beyond the above remarks, the reader is referred to \cite{gueant2013dealing} for a comprehensive analysis in the single-asset case. The complex interplay between price risk and liquidity risk expressed through $\Gamma=D_+^{-\frac 1 2}(D_+^{\frac 12}\Sigma D_+^{\frac 12})^{\frac 12} D_+^{-\frac 1 2}$ makes the sensitivity analysis less obvious in the multi-asset case.\\

\section{Beyond the quadratic approximation: towards a correction term} \label{sec:perturbative}

In Section \ref{sec:generaleq}, we approximated the Hamiltonian functions by quadratic functions in order to ``approximate'' the Hamilton-Jacobi equation characterizing the value function and then approximate the value function itself. To go further, we can consider a perturbative approach around our quadratic approximation. This means that we regard the real Hamiltonian functions as small perturbations of the quadratic functions used to approximate them and consider then a first order approximation (the zero-th order approximation being then that obtained in Section \ref{sec:generaleq}).\\

Formally, writing
$$H_\xi^{i,b}(p) = \check{H}^{i,b}(p) + \epsilon h^{i,b}(p), \quad
  H_\xi^{i,a}(p) = \check{H}^{i,a}(p) + \epsilon h^{i,a}(p), \quad
  \text{and } \quad \theta(t,q) = \check{\theta}(t,q) + \epsilon \eta(t,q),$$
and plugging these expressions in Eq. \eqref{sec2:thetagen} in the limit case where $Q^i = +\infty$ for all $i \in \{1, \ldots, d\}$, we obtain
\begingroup
\allowdisplaybreaks
\begin{eqnarray*}
0 &=& \partial_t \theta(t,q) - \frac 12 \gamma q^\intercal \Sigma q + \sum_{i=1}^d z^i H^{i,b}_\xi \left( \frac{\theta(t,q) - \theta(t,q+z^ie^i)}{z^i} \right) + \sum_{i=1}^d z^i H^{i,a}_\xi \left( \frac{\theta(t,q) - \theta(t,q-z^ie^i)}{z^i} \right)\\
    &=& \partial_t \check{\theta}(t,q) - \frac 12 \gamma q^\intercal \Sigma q + \sum_{i=1}^d z^i \check{H}^{i,b} \left( \frac{\check{\theta}(t,q) - \check{\theta}(t,q+z^ie^i)}{z^i} \right) + \sum_{i=1}^d z^i \check{H}^{i,a} \left( \frac{\check{\theta}(t,q) - \check{\theta}(t,q-z^ie^i)}{z^i} \right)\\
    &&+ \epsilon \left(\partial_t \eta(t,q) + \sum_{i=1}^d z^i h^{i,b} \left( \frac{\check{\theta}(t,q) - \check{\theta}(t,q+z^ie^i)}{z^i} \right) + \sum_{i=1}^d z^i h^{i,a} \left( \frac{\check{\theta}(t,q) - \check{\theta}(t,q-z^ie^i)}{z^i} \right)\right.\\
    && \qquad \qquad\quad+ \sum_{i=1}^d \check{H}^{i,b\prime} \left( \frac{\check{\theta}(t,q) - \check{\theta}(t,q+z^ie^i)}{z^i} \right) \left(\eta(t,q) - \eta(t,q+z^ie^i)\right)\\
    &&\left.\qquad\qquad\quad+ \sum_{i=1}^d \check{H}^{i,a\prime} \left( \frac{\check{\theta}(t,q) - \check{\theta}(t,q-z^ie^i)}{z^i} \right)\left(\eta(t,q) - \eta(t,q-z^ie^i)\right)\right) + o(\epsilon)\\
    &=& \epsilon \left(\partial_t \eta(t,q) + \sum_{i=1}^d z^i h^{i,b} \left( \frac{\check{\theta}(t,q) - \check{\theta}(t,q+z^ie^i)}{z^i} \right) + \sum_{i=1}^d z^i h^{i,a} \left( \frac{\check{\theta}(t,q) - \check{\theta}(t,q-z^ie^i)}{z^i} \right)\right.\\
    &&\qquad\qquad\quad+ \sum_{i=1}^d \left(-\check{H}^{i,b\prime} \left( \frac{\check{\theta}(t,q) - \check{\theta}(t,q+z^ie^i)}{z^i} \right)\right) \left(\eta(t,q+z^ie^i) - \eta(t,q)\right)\\
    &&\left.\qquad\qquad\quad+ \sum_{i=1}^d \left(-\check{H}^{i,a\prime} \left( \frac{\check{\theta}(t,q) - \check{\theta}(t,q-z^ie^i)}{z^i} \right)\right)\left(\eta(t,q-z^ie^i) - \eta(t,q)\right)\right) + o(\epsilon).
\end{eqnarray*}
\endgroup

Therefore,
\begin{eqnarray*}
0&=& \partial_t \eta(t,q) + \sum_{i=1}^d z^i h^{i,b} \left( \frac{\check{\theta}(t,q) - \check{\theta}(t,q+z^ie^i)}{z^i} \right) + \sum_{i=1}^d z^i h^{i,a} \left( \frac{\check{\theta}(t,q) - \check{\theta}(t,q-z^ie^i)}{z^i} \right)\\
&& + \sum_{i=1}^d \left(-\check{H}^{i,b\prime} \left( \frac{\check{\theta}(t,q) - \check{\theta}(t,q+z^ie^i)}{z^i} \right)\right) \left(\eta(t,q+z^ie^i) - \eta(t,q)\right)\\
&& + \sum_{i=1}^d \left(-\check{H}^{i,a\prime} \left( \frac{\check{\theta}(t,q) - \check{\theta}(t,q-z^ie^i)}{z^i} \right)\right)\left(\eta(t,q-z^ie^i) - \eta(t,q)\right),
\end{eqnarray*}
and we have the terminal condition $\eta(T,q) = 0$.\\

By Feynman-Kac representation theorem, we have
\begin{equation*}
\eta(t,q)\! =\! \mathbb{E}^{\check{\mathbb{P}}}\!\!\left[\!\int_t^T\! \left(\sum_{i=1}^d z^i h^{i,b} \left( \frac{\check{\theta}(s,q^{t,q}_{s-}) - \check{\theta}(s,q^{t,q}_{s-}+z^ie^i)}{z^i} \right)\! +\! \sum_{i=1}^d z^i h^{i,a} \left( \frac{\check{\theta}(s,q^{t,q}_{s-}) - \check{\theta}(s,q^{t,q}_{s-}-z^ie^i)}{z^i} \right)\right)\! ds\right]\!,
\end{equation*}
where under $\check{\mathbb{P}}$ the process $\left(q_s^{t,q}\right)_{s \in [t,T]}$ satisfies
$$dq^{t,q}_s = \sum_{i=1}^d z^i  (d\check{N}^{i,b}_{s} - d\check{N}^{i,a}_{s}) e^i \quad \text{and} \quad q^{t,q}_t = q,$$ with, for each $i \in \{1, \ldots, d\}$, $\check{N}^{i,b}$ and $\check{N}^{i,a}$ constructed like $N^{i,b}$ and $N^{i,a}$ but with respective intensities given at time $s$ by
$$-\check{H}^{i,b\prime} \left( \frac{\check{\theta}(s,q^{t,q}_{s-}) - \check{\theta}(s,q^{t,q}_{s-}+z^ie^i)}{z^i} \right) \quad \text{and} \quad -\check{H}^{i,a\prime} \left( \frac{\check{\theta}(s,q^{t,q}_{s-}) - \check{\theta}(s,q^{t,q}_{s-}-z^ie^i)}{z^i} \right).$$

In practice, it means that we can approximate $\theta(t,q)$ by $\check{\theta}(t,q)$ plus a correction term that takes the form of an expectation:
\begin{align*}
\theta(t,q) \simeq \check{\theta}(t,q) &+ \mathbb{E}^{\check{\mathbb{P}}}\left[\int_t^T \left(\sum_{i=1}^d z^i \left(H_\xi^{i,b} - \check{H}^{i,b}\right)\left( \frac{\check{\theta}(s,q^{t,q}_{s-}) - \check{\theta}(s,q^{t,q}_{s-}+z^ie^i)}{z^i} \right)\right.\right.\\
& \left.\left.+ \sum_{i=1}^d z^i \left(H_\xi^{i,a} - \check{H}^{i,a}\right)   \left( \frac{\check{\theta}(s,q^{t,q}_{s-}) - \check{\theta}(s,q^{t,q}_{s-}-z^ie^i)}{z^i} \right)\right) ds\right].
\end{align*}

Of course this new approximation is not a closed-form one. However, the correction term can be computed using a Monte-Carlo simulation for a specific $(t,q)$. In particular, it means that upon receiving a request for quote from a client and if time permits (which depends on asset class and market conditions), a market maker can perform a Monte-Carlo simulation to obtain an approximation of the value function at the relevant points to compute a quote that might account more accurately for the liquidity of the requested asset than a quote computed using the closed forms of Section \ref{subapproxgen}.\\

\section{A multi-asset market making model with additional features} \label{sec:extensions}

\subsection{A more general model}
\vspace{3mm}
In Section \ref{baseModel} we presented a multi-asset extension to the classical single-asset market making model of Avellaneda and Stoikov. This extension can itself be extended to encompass important features of OTC markets. In this section we extend our results to a more general multi-asset market making model with drift in prices to model the views of the market maker, client tiering, distributed requested sizes for each asset and each tier, and fixed transaction costs for each asset and each tier.\\

In terms of modeling, the addition of drifts to the price processes is straightforward. Formally, we assume that for each $i \in \{1,\ldots, d\}$, the dynamics of the price process $(S^i_t)_{ t \in \mathbb R_+}$ of asset $i$ is now given by
$$dS^i_t = \mu^i dt +  \sigma^i dW^i_t,$$
where $\sigma^i$ and $(W^i_t)_{t \in \mathbb R_+}$ are defined as in Section \ref{baseModel} and where $\mu^i$ is a real constant. In what follows, we denote by $\mu$ the vector $\mu = \left(\mu^1, \ldots, \mu^d \right)^\intercal$.\\

In OTC markets, market makers often divide their clients into groups, called tiers, for instance because they do not have the same commercial relationship with all clients or because the propensity to transact given a quote differs across clients. In particular, they can answer/stream different quotes to clients from different tiers.\footnote{There can also be tiers to proxy the existence of trading platforms with different clients and/or different costs.} Let us denote here by $N\in\mathbb N^*$ the number of such tiers.\\

In addition to introducing tiers, we can drop the assumption of constant request size per asset and consider instead that, for each asset and each tier, the size of the requests at the bid and at the ask are distributed according to known distributions.\\

Mathematically, the bid and ask quotes that the market maker propose are then modeled by the maps
$$S^{i,n,b} :\! (\omega,t,z) \in \Omega \times [0,T] \times \mathbb R_+^* \mapsto S^{i,n,b}_t(\omega,z) \in \mathbb R \!\!\! \quad \text{and} \!\!\! \quad S^{i,n,a} :\! (\omega,t,z) \in \Omega \times [0,T] \times \mathbb R_+^* \mapsto S^{i,n,a}_t(\omega,z) \in \mathbb R, $$
where $i \in\{1, \ldots, d\}$ is the index of the asset, $n \in \{1, \ldots, N\}$ is the index of the tier, and $z \in \mathbb R_+^*$ is the size of the request (in number of assets). In the same vein as in Section \ref{baseModel}, we introduce\footnote{$\omega$ is omitted in what follows.}   $$\delta^{i,n,b}_t(z) = S^i_t - S^{i,n,b}_t(z) \quad \text{and} \quad \delta^{i,n,a}_t(z) = S^{i,n,a}_t(z) - S^i_t,$$ and the maps $(\delta^{i,n,b}_t(.))_{t\in \mathbb R_+}$ and $(\delta^{i,n,a}_t(.))_{t\in \mathbb R_+}$ are assumed to be $\mathbb F$-predictable and bounded from below by a given constant $-\delta_{\infty }<0$.\footnote{This additional constraint of a fixed lower bound is just a technical one to be able to state theorems in the general case where request sizes are distributed (see \cite{bergault2019size}).} \\

With these new features in mind, we introduce for each asset $i \in \{1,\ldots, d\}$ and for each tier $n \in \{1,\ldots,N\}$ the processes $(N^{i,n,b}_t)_{t \in \mathbb R_+}$ and $(N^{i,n,a}_t)_{t \in \mathbb R_+}$ modeling the number of transactions in asset $i$ with clients from tier~$n$ at the bid and at the ask, respectively. They are $\mathbb R_+^*$-marked point processes, with respective intensity kernels $(\lambda^{i,n,b}_t(dz))_{t \in \mathbb R_+^*}$ and $(\lambda^{i,n,a}_t(dz))_{t \in \mathbb R_+^*}$ given by $$\lambda^{i,n,b}_t(dz) = \Lambda^{i,n,b}(\delta^{i,n,b}_t(z)) \mathds{1}_{\{q^i_{t-} + z \le Q^i \}} \nu^{i,n,b}(dz)\!\! \quad \text{and} \!\! \quad \lambda^{i,n,a}_t(dz) = \Lambda^{i,n,a}(\delta^{i,n,a}_t(z)) \mathds{1}_{\{q^i_{t-} - z \ge -Q^i \}}  \nu^{i,n,a}(dz), $$
where $\nu^{i,n,b}$ and $\nu^{i,n,a}$ are the two probability measures representing the distribution of the requested sizes at the bid and at the ask respectively, for asset $i$ and tier $n$, and where $\Lambda^{i,n,b}$ and $\Lambda^{i,n,a}$ satisfy the same assumptions as those satisfied by the intensity functions of Section \ref{baseModel}.\\

For asset $i \in \{1,\ldots , d\}$, the resulting inventory $(q^i_t)_{t \in \mathbb R_+}$ has dynamics $$dq^i_t = \sum_{n=1}^N \int_{\mathbb R_+^*} z N^{i,n,b}(dt,dz) - \sum_{n=1}^N \int_{\mathbb R_+^*} z N^{i,n,a}(dt,dz),$$ where for each tier $n \in \{1,\ldots, N\}  ,$ $N^{i,n,b}(dt,dz)$ and $N^{i,n,a}(dt,dz)$ are the random measures associated with the processes $(N^{i,n,b}_t)_{t \in \mathbb R_+}$ and $(N^{i,n,a}_t)_{t \in \mathbb R_+}$, respectively.\\

Finally, we consider the addition of fixed transaction costs.\footnote{Proportional transaction costs can be considered in the initial model through shifts in the intensity functions.} For that purpose, we introduce for each asset $i \in \{1,\ldots, d\}$ and for each  tier $n \in \{1,\ldots,N\}$ two real numbers $c^{i,n,b}$ and $c^{i,n,a}$ modelling the fixed cost of a transaction in asset $i$ with a client from tier $n$, at the bid and at the ask, respectively.\\

The resulting cash process $(X_t)_{t \in \mathbb R_+}$ has, consequently, the following dynamics:
$$dX_t = \sum_{i=1}^d \sum_{n=1}^N \int_{\mathbb{R}_+^*} \left[ \left( \delta^{i,n,b}_t(z)z - c^{i,n,b} \right) N^{i,n,b}(dt,dz) +   \left( \delta^{i,n,a}_t(z)z - c^{i,n,a} \right) N^{i,n,a}(dt,dz) \right] -  \sum_{i=1}^d  S^{i}_t dq^i_t.$$

\subsection{The Hamilton-Jacobi equation}
\vspace{3mm}
In this new setting, one can again show that the two optimization problems introduced in Section \ref{baseModel} boil down to the resolution of a Hamilton-Jacobi equation of the form
\begin{align}
\label{HJfeat}
 0 = &\ \partial_t \theta(t,q) + \mu^\intercal q - \frac{\gamma}{2} q^\intercal \Sigma q \\
 &\ + \sum_{i=1}^d \sum_{n=1}^N \int_{\mathbb R_+^*} \mathds{1}_{\{ q^i + z \le Q^i\}} z H^{i,n,b}_{\xi} \left(z,  \frac{\theta(t,q) - \theta(t,q+ze^i) + c^{i,n,b}}{z} \right) \nu^{i,n,b}(dz) \nonumber\\
 & \ + \sum_{i=1}^d \sum_{n=1}^N \int_{\mathbb R_+^*} \mathds{1}_{\{ q^i - z \ge -Q^i \}} z H^{i,n,a}_{\xi} \left(z, \frac{\theta(t,q) - \theta(t,q-ze^i)  + c^{i,n,a}}{z} \right) \nu^{i,n,a}(dz),\nonumber
\end{align}
with terminal condition \begin{align}
    \label{sec5:thetagenCT}
    \theta(T,q) = 0,
\end{align} where $\xi = \gamma$ in the case of Model A and $\xi=0$ in the case of Model B, and where the functions $H^{i,n,b}_{\xi}$ and $H^{i,n,a}_{\xi}$ are defined by
\[ H^{i,n,b}_\xi(z,p):=
	\begin{cases}
		\underset{ \delta>-\delta_{\infty}}{\sup}\frac{\Lambda^{i,n,b}(\delta)}{\xi z}(1-\exp(-\xi z(\delta-p)))&\mbox{if } \xi>0,\\
		\underset{ \delta>-\delta_{\infty}}{\sup}\Lambda^{i,n,b}(\delta)(\delta-p)&\mbox{if }\xi=0
	\end{cases} \]
and
\[ H^{i,n,a}_\xi(z,p):=
	\begin{cases}
		\underset{\delta>-\delta_{\infty}}{\sup}\frac{\Lambda^{i,n,a}(\delta)}{\xi z}(1-\exp(-\xi z(\delta-p)))&\mbox{if } \xi>0,\\
		\underset{ \delta>-\delta_{\infty}}{\sup}\Lambda^{i,n,a}(\delta)(\delta-p)&\mbox{if }\xi=0.
	\end{cases} \]

\begin{remark}
When $\xi=0$, the dependence in $z$ of the Hamiltonian functions vanishes. Nevertheless, we keep the variable $z$ for the sake of consistency.\\
\end{remark}

Following a method similar to that developed in \cite{bergault2019size}, we can show that, for a given $\xi \ge 0,$ there exists a unique bounded function $\theta:[0,T]\times \prod_{i=1}^d [-Q^i, Q^i]\to \Rr$, $C^1$ in time, solution of Eq.~\eqref{HJfeat} with terminal condition \eqref{sec5:thetagenCT}.\\

Moreover, under the (mild) assumption that the measures $\nu^{i,n,b}$ and $\nu^{i,n,a}$ have moments of order $2$, a classical verification argument enables to go from $\theta$ to optimal controls for both Model A and Model B. The optimal quotes as functions of $\theta$ are given by the two theorems that follow.\\

In the case of Model A, the result is the following:
\begin{theorem}
\label{optquotes5}
	Let us consider the solution $\theta$ of Eq.~\eqref{HJfeat} with terminal condition \eqref{sec5:thetagenCT}, for $\xi=\gamma$.\\
	
	Then, for $i \in \{1, \ldots, d\}$ and $n \in \{1,\ldots,N\}$, the optimal bid and ask quotes as functions of the trade size $z$, $S^{i,n,b}_t(z) = S^i_t - \delta^{i,n,b*}_t(z)$ and $S^{i,n,a}_t(z) = S^i_t + \delta^{i,n,a*}_t(z)$ in Model A are characterized by\begin{align*}
	\begin{split}
	\delta^{i,n,b*}_t(z) &= \tilde{\delta}^{i,n,b*}_\gamma\left(z,\frac{\theta(t,q_{t-}) - \theta(t,q_{t-}+z e^i)  + c^{i,n,b}}{z}\right) \quad \text{for } q_{t-}+z e^i \in \prod_{j=1}^d[-Q^j,Q^j],\\
	\delta^{i,n,a*}_t(z) &= \tilde{\delta}^{i,n,a*}_\gamma\left(z,\frac{\theta(t,q_{t-}) - \theta(t,q_{t-}-z e^i) + c^{i,n,a}}{z}\right) \quad \text{for } q_{t-}-z e^i \in \prod_{j=1}^d[-Q^j,Q^j],\\
	\end{split}
	\end{align*}where the functions $\tilde{\delta}^{i,n,b*}_\gamma(\cdot,\cdot)$ and $\tilde{\delta}_\gamma^{i,n,a*}(\cdot, \cdot)$ are defined by
	\begin{align*}
	\tilde{\delta}^{i,n,b*}_\gamma(z,p) &= {\Lambda^{i,n,b}}^{-1}\left(\gamma z H^{i,n,b}_{\gamma}(z,p) - {\partial_p H_{\gamma}^{i,n,b}}(z,p)\right) \vee (-\delta_\infty) ,\\
	\tilde{\delta}^{i,n,a*}_\gamma(z,p) &= {\Lambda^{i,n,a}}^{-1}\left(\gamma z H^{i,n,a}_{\gamma}(z,p) - {\partial_pH_{\gamma}^{i,n,a}}(z,p)\right) \vee (-\delta_\infty).
	\end{align*}	
\end{theorem}
For Model B, the result is the following:
\begin{theorem}
\label{optquotesbis5}
	Let us consider the solution $\theta$ of Eq.~\eqref{HJfeat} with terminal condition \eqref{sec5:thetagenCT}, for $\xi=0$.\\
	
	Then, for $i \in \{1, \ldots, d\}$ and $n \in \{1,\ldots,N\}$, the optimal bid and ask quotes as functions of the trade size $z$, $S^{i,n,b}_t(z) = S^i_t - \delta^{i,n,b*}_t(z)$ and $S^{i,n,a}_t(z) = S^i_t + \delta^{i,n,a*}_t(z)$ in  Model B are characterized by\begin{align*}
		\begin{split}
	\delta^{i,n,b*}_t(z) &= \tilde{\delta}^{i,n,b*}_0\left(z,\frac{\theta(t,q_{t-}) - \theta(t,q_{t-}+z e^i)  + c^{i,n,b}}{z}\right) \quad \text{for } q_{t-}+z e^i \in \prod_{j=1}^d[-Q^j,Q^j],\\
	 \delta^{i,n,a*}_t(z) &= \tilde{\delta}^{i,n,a*}_0\left(z,\frac{\theta(t,q_{t-}) - \theta(t,q_{t-}-z e^i)  + c^{i,n,a}}{z}\right) \quad \text{for } q_{t-}-z e^i \in \prod_{j=1}^d[-Q^j,Q^j],\\
	 \end{split}
	\end{align*}where the functions $\tilde{\delta}^{i,n,b*}_0(\cdot,\cdot)$ and $\tilde{\delta}_0^{i,n,a*}(\cdot,\cdot)$ are defined by
	$$\tilde{\delta}^{i,n,b*}_0(z,p) = {\Lambda^{i,n,b}}^{-1}\left(- {\partial_pH_{0}^{i,n,b}}(z,p)\right)\vee (-\delta_\infty) \text{ and } \tilde{\delta}^{i,n,a*}_0(z,p) = {\Lambda^{i,n,a}}^{-1}\left(- {\partial_pH_{0}^{i,n,a}}(z,p)\right)\vee (-\delta_\infty).$$
\end{theorem}

\subsection{Quadratic approximation}
\vspace{3mm}
As before, let us replace for all $i \in \{1, \ldots, d\}$ and $n\in \{1,\ldots,N\}$, the Hamiltonian functions $H_\xi^{i,n,b}$ and $H_ \xi^{i,n,a}$ by the functions
$$\check{H}^{i,n,b}: (z,p) \mapsto  \alpha^{i,n,b}_0(z) + \alpha^{i,n,b}_1(z) p + \frac 12 \alpha^{i,n,b}_2(z) p^2 \textrm{ and } \check{H}^{i,n,a}: (z,p) \mapsto  \alpha^{i,n,a}_0(z) + \alpha^{i,n,a}_1(z) p + \frac 12 \alpha^{i,n,a}_2(z) p^2.$$

\begin{remark} \label{remp0gen}
Here, $\alpha^{i,n,b}_j$ and $\alpha^{i,n,a}_j$ (for $j \in \{0,1,2\}$) are functions of $z.$ A natural choice for the functions $(\check{H}^{i,n,b})_{(i,n) \in \{1, \ldots, d\} \times  \{1,\ldots,N\}}$ and $(\check{H}^{i,n,a})_{(i,n) \in \{1, \ldots, d\} \times \{1,\ldots,N\}}$ derives from Taylor expansions around $p=0$. In that case,
$$\forall i \in \{1, \ldots, d\}, \forall n\in \{1,\ldots,N\}, \forall j \in \{0,1,2\},\quad  \alpha^{i,n,b}_j(z) = \partial^j_p{H^{i,n,b}_\xi}(z,0) \quad \textrm{and} \quad \alpha^{i,n,a}_j = \partial^j_p{H^{i,n,a}_\xi}(z,0).$$
\end{remark}

If we consider the limit case where $ Q^i = +\infty$ for all $i \in \{1, \ldots, d\}$, Eq. \eqref{HJfeat} then becomes
\begin{align}
 0 = &\ \partial_t \check{\theta}(t,q) + \mu^\intercal q - \frac{\gamma}{2} q^\intercal \Sigma q + \sum_{i=1}^d \sum_{n=1}^N  \left(\int_{\mathbb R_+^*} z\alpha^{i,n,b}_0(z) \nu^{i,n,b}(dz) + \int_{\mathbb R_+^*} z\alpha^{i,n,a}_0(z) \nu^{i,n,a}(dz) \right) \nonumber\\
 &\ + \sum_{i=1}^d \sum_{n=1}^N \Bigg(\int_{\mathbb R_+^*}  \alpha^{i,n,b}_1(z) \left(  \check{\theta}(t,q) - \check{\theta}(t,q+ze^i)  + c^{i,n,b}\right) \nu^{i,n,b}(dz) \nonumber\\
 & \qquad \qquad \qquad \qquad +  \int_{\mathbb R_+^*}\alpha^{i,n,a}_1(z) \left( \check{\theta}(t,q) - \check{\theta}(t,q-ze^i)  + c^{i,n,a} \right) \nu^{i,n,a}(dz)\Bigg)\nonumber\\
 &\ +\frac 12  \sum_{i=1}^d \sum_{n=1}^N \Bigg(\int_{\mathbb R_+^*} \frac 1z  \alpha^{i,n,b}_2(z) \left(  \check{\theta}(t,q) - \check{\theta}(t,q+ze^i)  + c^{i,n,b}\right)^2 \nu^{i,n,b}(dz) \nonumber\\
 & \qquad \qquad \qquad \qquad +  \int_{\mathbb R_+^*} \frac 1z \alpha^{i,n,a}_2(z) \left( \check{\theta}(t,q) - \check{\theta}(t,q-ze^i)  + c^{i,n,a} \right)^2 \nu^{i,n,a}(dz)\Bigg),\label{HJfeatapprox}
\end{align}

with terminal condition
\begin{align}
\label{genCT}
\check \theta (T,q) = 0.
\end{align}

Using the same ansatz as in Section \ref{sec:generaleq}, we obtain the following result (we omit the proof as it follows the same logic as for that of Proposition \ref{sec3:ansatz}):

\begin{proposition}\label{prop:GENansatz}
Let us introduce for all $i \in \{1, \ldots, d\}, n \in \{1,\ldots,N\}, j \in \{0, 1, 2\}, k \in \mathbb{N}$, the following constants:

$$\Delta^{i,n,b}_{j,k} = \int_{\mathbb R_+^*}\!\! z^k \alpha^{i,n,b}_j(z) \nu^{i,n,b}(dz) \quad \text{and} \quad \Delta^{i,n,a}_{j,k} = \int_{\mathbb R_+^*}\!\! z^k \alpha^{i,n,a}_j(z) \nu^{i,n,a}(dz),$$
$$V^{b}_{j,k} = \left(\sum_{n=1}^N{\Delta^{1,n,b}_{j,k}},\ldots,\sum_{n=1}^N{\Delta^{d,n,b}_{j,k}} \right)^\intercal \quad \text{and} \quad V^{a}_{j,k} = \left(\sum_{n=1}^N{\Delta^{1,n,a}_{j,k}},\ldots,\sum_{n=1}^N{\Delta^{d,n,a}_{j,k}} \right)^\intercal,$$
$${\tilde{V}^{b}_{j,k}} = \left(\sum_{n=1}^N c^{1,n,b}{\Delta^{1,n,b}_{j,k}},\ldots,\sum_{n=1}^N c^{d,n,b}{\Delta^{d,n,b}_{j,k}} \right)^\intercal \quad \text{and} \quad {\tilde{V}^{a}_{j,k}} = \left(\sum_{n=1}^N c^{1,n,a}{\Delta^{1,n,a}_{j,k}},\ldots,\sum_{n=1}^N c^{d,n,a}{\Delta^{d,n,a}_{j,k}} \right)^\intercal ,$$
$$\chi^{b}_{j,k} = \sum_{i=1}^d \sum_{n=1}^N{\Delta^{i,n,b}_{j,k}} \quad \text{and} \quad \chi^{a}_{j,k} = \sum_{i=1}^d \sum_{n=1}^N{\Delta^{i,n,a}_{j,k}} ,$$
$${\tilde{\chi}^{b}_{j,k}} = \sum_{i=1}^d \sum_{n=1}^N c^{i,n,b}{\Delta^{i,n,b}_{j,k}} \quad \text{and} \quad {\tilde{\chi}^{a}_{j,k}} = \sum_{i=1}^d \sum_{n=1}^N c^{i,n,a}{\Delta^{i,n,a}_{j,k}},$$
$${\hat{\chi}^{b}_{j,k}} = \sum_{i=1}^d \sum_{n=1}^N (c^{i,n,b})^2 {\Delta^{i,n,b}_{j,k}} \quad \text{and} \quad {\hat{\chi}^{a}_{j,k}} = \sum_{i=1}^d \sum_{n=1}^N (c^{i,n,a})^2 {\Delta^{i,n,a}_{j,k}},$$
and
$$D^b_{j,k} = \textup{diag} \left(\sum_{n=1}^N{\Delta^{1,n,b}_{j,k}},\ldots,\sum_{n=1}^N{\Delta^{d,n,b}_{j,k}}\right) \quad \text{and} \quad D^a_{j,k} = \textup{diag} \left(\sum_{n=1}^N{\Delta^{1,n,a}_{j,k}},\ldots,\sum_{n=1}^N{\Delta^{d,n,a}_{j,k}} \right).$$

Let us consider three differentiable functions $A: [0,T] \to S_d^{+}$, $B: [0,T] \to \mathbb{R}^d$, and $C: [0,T] \to \mathbb R$ solutions of the system of ordinary differential equations
\begin{equation}\label{system:ABCgeneral}
\begin{cases}
\displaystyle
A'(t) =&\ 2A(t) \left(D^b_{2,1} + D^a_{2,1} \right) A(t)  - \frac 12 \gamma \Sigma , \\
B'(t) =&\ \mu + 2A(t) \left( {V^{b}_{1,1}} - {V^{a}_{1,1}} \right) + 2 A(t) \left(D^{b}_{2,2} - D^{a}_{2,2}  \right) \mathcal D (A(t)) \\
    &\ + 2 A(t) \left(D^{b}_{2,1} + D^{a}_{2,1} \right) B(t)  + 2 A(t) \left( {\tilde{V}^{b}_{2,0}} - {\tilde{V}^{a}_{2,0}} \right),\\
C'(t) = &\ \textrm{Tr}\left( D^b_{0,1} + D^a_{0,1} \right) +\textrm{Tr}\left(\left(D^b_{1,2} + D^a_{1,2}\right) A(t)\right)  + \left( {V^{b}_{1,1}} - {V^{a}_{1,1}} \right)^\intercal B(t)\\
    &\  + \left( \tilde{\chi}^{b}_{1,0} + \tilde{\chi}^{a}_{1,0} \right) + \frac 12 \mathcal{D}(A(t))^\intercal \left(D^b_{2,3} + D^a_{2,3}\right) \mathcal{D}(A(t))\\
    &\ + B(t)^\intercal \left(D^b_{2,2} - D^a_{2,2}\right) \mathcal{D}(A(t)) + \left( {\tilde{V}^{b}_{2,1}} + {\tilde{V }^{a}_{2,1}} \right)^\intercal \mathcal D(A(t)) \\
    &\  + \frac 12 B(t)^\intercal \left({D^{b}_{2,1}} + {D^{a}_{2,1}} \right) B(t)  + \left(  {\tilde{V}^{b}_{2,0}} - {\tilde{V}^{a}_{2,0}} \right)^\intercal B(t)  \\
&\  + \frac 12 \left(\hat{\chi}^{b}_{2,0} + \hat{\chi}^{a}_{2,0} \right).
\end{cases}
\end{equation}
with terminal conditions
\begin{equation}
\label{ABCgenTC}
    A(T) = 0, B(T) = 0, \textrm{ and } C(T)=0.
\end{equation}

Then $\check{\theta}: (t,q) \in [0,T] \times \mathbb R^d \mapsto -q^\intercal A(t) q - q^\intercal B(t) -  C(t)$ is solution of Eq. \eqref{HJfeatapprox} with terminal condition~\eqref{genCT}.
\end{proposition}

We can now solve \eqref{system:ABCgeneral} with terminal conditions \eqref{ABCgenTC} in closed form. This is the purpose of the following proposition whose proof is omitted (see Proposition \ref{prop:ode_solution} for a similar proof).\\

\begin{proposition}\label{prop:GENode_solution}
Assume $\sum_{n=1}^N \Delta^{i,n, b}_{2,1} + \Delta^{i,n, a}_{2,1} > 0$ for all $i \in \{ 1, \ldots, d
\}$. The system of ODEs
\eqref{system:ABCgeneral} with terminal conditions \eqref{ABCgenTC} admits the
unique solution
\begin{align*}
A(t) & = \frac{1}{2} D^{-\frac 12}_+ \widehat{A}
    \left( e^{ \widehat{A} (T - t)} - e^{-\widehat{A} (T - t)} \right)
    \left( e^{ \widehat{A} (T - t)} + e^{-\widehat{A} (T - t)} \right)^{-1} D_+^{-\frac 12}, \\
B(t) & = -2 e^{-2 \int_t^T A(u) D_+ \,du} \int_t^T
    e^{2 \int_s^T A(u) D_+ \,du}\left(\frac 12 \mu +  A(s) \left(V_- +  \tilde V_- + D_- \mathcal D(A(s)) \right) \right) ds, \\
C(t) & =  -\textrm{Tr}\left(D^b_{0,1} + D^a_{0,1}\right) (T - t)
       -  \textrm{Tr}\left(\left(D^b_{1,2} + D^a_{1,2}\right) \int_t^T A(s) ds \right)
       -  V_-^\intercal  \int_t^T B(s) ds \\
     & \quad  - \frac 12 \int_t^T \mathcal{D}(A(s))^\intercal \left(D^b_{2,3}
         + D^a_{2,3}\right) \mathcal{D}(A(s)) ds
       - \frac 12 \int_t^T B(s)^\intercal D_+ B(s) ds \\
     & \quad - \int_t^T B(s)^\intercal D_- \mathcal{D}(A(s)) ds - \left( \tilde{\chi}^{b}_{1,0} + \tilde{\chi}^{a}_{1,0} \right) (T-t) - \frac 12 \left(\hat{\chi}^{b}_{2,0} + \hat{\chi}^{a}_{2,0} \right)(T-t) \\
     & \quad - \left( {\tilde{V}^{b}_{2,1}} + {\tilde{V }^{a}_{2,1}} \right)^\intercal \int_t^T \mathcal D(A(s))ds,
\end{align*}
where
\[ D_+ = D^b_{2,1} + D^a_{2,1}, \quad
   D_- = D_{2,2}^b - D_{2,2}^a, \quad
   V_- = V_{1,1}^b - V_{1,1}^a, \quad \tilde V_- =  \tilde V^b_{2,0} - \tilde V^a_{2,0},\text{ and }
   \widehat{A}
   = \sqrt{\gamma} \left( D_+^{\frac 12} \Sigma D_+^{\frac 12} \right)^{\frac 12}.\]
\end{proposition}

Now, using the same method as in Section \ref{sec:generaleq}, we get the following asymptotic results:

\begin{proposition}\label{prop:GENode_asymp_solution}
Let $(A,B,C)$ be the solution of the system of ODEs \eqref{system:ABCgeneral} with terminal conditions \eqref{ABCgenTC}.\\

If $D_+^{\frac 12} \mu \in \text{Im}(\widehat A)$, then,
\begingroup
\allowdisplaybreaks
\begin{align*}
A(0) & \stackrel{T\to+\infty}{\longrightarrow} \frac{1}{2}\sqrt{\gamma}\Gamma, \\
B(0) & \stackrel{T\to+\infty}{\longrightarrow} - D_+^{-\frac 12} \widehat A^+ D_+^{\frac 12} \mu - D_+^{-\frac 12}\widehat A \widehat A ^+ D_+^{-\frac 12} \left( V_- + \tilde V_-
     + \frac{1}{2}\sqrt{\gamma}  D_- \mathcal D(\Gamma) \right),\\
\frac{C(0)}T & \stackrel{T\to+\infty}{\longrightarrow} -\textrm{Tr}\left(D^b_{0,1} + D^a_{0,1}\right) -\frac{1}{2}\sqrt{\gamma} \textrm{Tr}\left(\left(D^b_{1,2} + D^a_{1,2}\right)\Gamma\right) + V_-^\intercal D_+^{-\frac 12} \widehat A^+ D_+^{\frac 12} \mu \\
&\qquad + V_-^\intercal D_+^{-\frac 12}\widehat{A}\widehat{A}^+ D_+^{-\frac 12}\!\left(V_- + \tilde V_-
     + \frac{1}{2}\sqrt{\gamma}  D_- \mathcal D(\Gamma)\right)\!\!-\!\frac 18 \gamma \mathcal{D}(\Gamma)^\intercal\!\!\left(D^b_{2,3}\! +\! D^a_{2,3}\right) \mathcal{D}(\Gamma)\\
     &\qquad - \frac 12 \mu^\intercal D_+^{\frac 12}\widehat A ^+ \widehat A^+ D_+^{\frac 12} \mu - \mu^\intercal D_+^{\frac 12} \widehat A ^+ D_+^{-\frac 12} \left(V_- + \tilde V_-
     + \frac{1}{2}\sqrt{\gamma}  D_- \mathcal D(\Gamma)\right) \\
          &\qquad - \frac 12 \left(V_- + \tilde V_-
     + \frac{1}{2}\sqrt{\gamma}  D_- \mathcal D(\Gamma)\right)^\intercal\! D_+^{-\frac 12}\widehat{A}\widehat{A}^+ D_+^{-\frac 12} \left(V_- + \tilde V_-
     + \frac{1}{2}\sqrt{\gamma}  D_- \mathcal D(\Gamma)\right)\\
     &\qquad + \frac{1}{2}\sqrt{\gamma} \mu^\intercal  D_+^{\frac 12} \widehat A^+ D_+^{-\frac 12} D_- \mathcal D(\Gamma)  + \frac{1}{2}\sqrt{\gamma} \left(V_- + \tilde V_-
     + \frac{1}{2}\sqrt{\gamma}  D_- \mathcal D(\Gamma)\right)^\intercal D_+^{-\frac 12}\widehat{A}\widehat{A}^+ D_+^{-\frac 12}  D_- \mathcal{D}(\Gamma) \\
     &\qquad - \left(\hat{\chi}^{b}_{2,0} + \hat{\chi}^{a}_{2,0} \right) - \frac 12 \left(\hat{\chi}^{b}_{2,0} + \hat{\chi}^{a}_{2,0} \right) - \frac 12 \sqrt{\gamma} \left( {\tilde{V}^{b}_{2,1}} + {\tilde{V }^{a}_{2,1}} \right)^\intercal \mathcal{D}(\Gamma),
\end{align*}
\endgroup
where $\Gamma = D_+^{-\frac 1 2} \left( D_+^{\frac 12} \Sigma D_+^{\frac 12} \right)^{\frac 12}
            D_+^{-\frac 1 2}$ and $\widehat{A}^+$ is the Moore-Penrose generalized inverse of
$\widehat{A}$.\\
\end{proposition}

\begin{remark}
  The assumption $D_+^{\frac 12} \mu \in \text{Im}(\widehat A)$ is satisfied when $\mu = 0$ or when $\Sigma$ is invertible. If this assumption is not satisfied, then it can be shown that $\frac{B(0)}{T} \stackrel{T\to+\infty}{\longrightarrow} - D_+^{-\frac 12} \widehat{A}^+ \widehat{A} D_+^{\frac 12} \mu$. In particular, there is no constant asymptotic approximation of the quotes. In fact, if the assumption $D_+^{\frac 12} \mu \in \text{Im}(\widehat A)$ is not satisfied, the market maker may have an incentive to propose very good quotes to clients in order to build portfolios bearing positive return at no risk.\\
\end{remark}

\subsection{From value functions to heuristics and quotes}
\vspace{3mm}
\subsubsection{Quotes: the general case}
\vspace{3mm}
The greedy quoting strategy associated with our closed-form proxy of the value function leads to the following quotes for all $i \in \{1, \ldots, d\}$ and $n \in \{1,\ldots, N\}$:

\begin{eqnarray*}
	\check\delta^{i,n,b}_t(z) &=& \tilde{\delta}^{i,n,b*}_\xi\left(z,\frac{\check\theta(t,q_{t-}) - \check\theta(t,q_{t-}+z e^i) + c^{i,n,b}}{z}\right)\\
 &=& \tilde{\delta}^{i,n,b*}_\xi\left(z,2q_{t-}^\intercal A(t) e^i + z {e^i}^\intercal A(t) e^i + {e^i}^\intercal B(t) + \frac{c^{i,n,b}}{z} \right),\\
	\check\delta^{i,n,a}_t(z) &=& \tilde{\delta}^{i,n,a*}_\xi\left(z,\frac{\check\theta(t,q_{t-}) - \check\theta(t,q_{t-}-z e^i) + c^{i,n,a}}{z}\right)\\
&=& \tilde{\delta}^{i,n,a*}_\xi\left(z,-2q_{t-}^\intercal A(t) e^i + z {e^i}^\intercal A(t) e^i - {e^i}^\intercal B(t)  + \frac{c^{i,n,a}}{z} \right),
	\end{eqnarray*}
where $\tilde{\delta}^{i,n,b*}_\xi$ and $\tilde{\delta}^{i,n,a*}_\xi$ are given in Theorems \ref{optquotes5} and \ref{optquotesbis5} for $\xi = \gamma$ and $\xi = 0$ respectively (depending on whether one considers Model A or Model B).\\

The asymptotic regime exhibited in the above paragraphs can then serve to obtain the following simple closed-form approximations:
\begin{align}
	\breve\delta^{i,n,b}_t(z) &=\tilde{\delta}^{i,n,b*}_\xi \bigg( z,\sqrt\gamma q_{t-}^\intercal \Gamma e^i + \frac 12 \sqrt\gamma z {e^i}^\intercal \Gamma e^i - {e^i}^\intercal D_+^{-\frac 12} \widehat A^+ D_+^{\frac 12} \mu \label{asymptbgen}\\
	& \qquad \qquad \qquad \qquad \qquad - {e^i}^\intercal D_+^{-\frac 12}\widehat{A}\widehat{A}^+ D_+^{-\frac 12} \left(V_- + \tilde V_-
     + \frac{1}{2}\sqrt{\gamma}  D_- \mathcal D(\Gamma) \right) +\frac{c^{i,n,b}}{z} \bigg),\nonumber \\
	\breve\delta^{i,n,a}_t(z) &= \tilde{\delta}^{i,n,a*}_\xi \bigg( z,-\sqrt\gamma q_{t-}^\intercal \Gamma e^i + \frac 12 \sqrt\gamma z {e^i}^\intercal \Gamma e^i +  {e^i}^\intercal D_+^{-\frac 12} \widehat A^+ D_+^{\frac 12} \mu \label{asymptagen}\\
	& \qquad \qquad \qquad \qquad \qquad   +{e^i}^\intercal D_+^{-\frac 12}\widehat{A}\widehat{A}^+ D_+^{-\frac 12} \left(V_- + \tilde V_-
     + \frac{1}{2}\sqrt{\gamma}  D_- \mathcal D(\Gamma) \right) +\frac{c^{i,n,a}}{z} \bigg). \nonumber
	\end{align}

If we assume that for all $i \in \{1, \ldots, d\}$ and for all $n \in \{1,\ldots,N\}$ we have $\nu^{i,n,b} = \nu^{i,n,a}$ and $\Lambda^{i,n,b} = \Lambda^{i,n,a}$, then $\forall i \in \{1, \ldots, d\}, \forall n \in \{1,\ldots,N\}, H^{i,n,b} = H^{i,n,a}$, and it is thus natural to chose symmetric approximations of the Hamiltonian functions, i.e. $\forall i \in \{1, \ldots, d\}, \forall n \in \{1,\ldots,N\}, \check{H}^{i,n,b} = \check{H}^{i,n,a}$. In that case, \eqref{asymptbgen} and \eqref{asymptagen} simplify into
\begin{eqnarray}
	\breve\delta^{i,n,b}_t(z) &=& \tilde{\delta}^{i,n,b*}_\xi\left(z,\sqrt\gamma q_{t-}^\intercal \Gamma e^i + \frac 12 \sqrt\gamma z {e^i}^\intercal \Gamma e^i -   {e^i}^\intercal D_+^{-\frac 12} \widehat A^+ D_+^{\frac 12} \mu +\frac{c^{i,n,b}}{z} \right),\label{asymptbsymgen}\\
	\breve\delta^{i,n,a}_t(z) &=& \tilde{\delta}^{i,n,a*}_\xi\left(z,-\sqrt\gamma q_{t-}^\intercal \Gamma e^i + \frac 12 \sqrt\gamma z {e^i}^\intercal \Gamma e^i +  {e^i}^\intercal D_+^{-\frac 12} \widehat A^+ D_+^{\frac 12} \mu+\frac{c^{i,n,a}}{z} \right).\label{asymptasymgen}
	\end{eqnarray}

All these approximations of the quotes can be used directly or as a starting point in iterative methods designed to compute the optimal quotes (policy iteration algorithms, actor-critic algorithms, etc.).\\

\subsubsection{Quotes: the case of symmetric exponential intensities}
\vspace{3mm}
If we assume that for all $i \in \{1, \ldots, d\}$ and for all $n \in \{1,\ldots,N\}$ we have $\nu^{i,n,b} = \nu^{i,n,a} =: \nu^{i,n}$ and intensity functions given by
\[    \Lambda^{i,n,b} (\delta) = \Lambda^{i,n,a} (\delta) = A^{i,n} e^{-k^{i,n} \delta}, \quad A^{i,n}, k^{i,n} > 0,\]
then (see \cite{gueant2017optimal}), in the limit case where $\delta_\infty = +\infty$ the Hamiltonian functions are given, for all $i \in \{1, \ldots, d\}$ and $n\in \{1,\ldots,N\}$, by
\[H^{i,n,b}_\xi (z,p)
   =  H^{i,n,a}_\xi (z,p) = \frac{A^{i,n}}{k^{i,n}} C^{i,n}_\xi(z) \exp(-k^{i,n} p), \]
 where
 \[ C_\xi^{i,n}(z) =
 \begin{cases}
 \left(1+\frac{\xi z}{k^{i,n}}\right)^{-\left(1+\frac{k^{i,n}}{\xi z}\right)}&\mbox{if }\xi>0\\
 e^{-1}&\mbox{if }\xi=0,
 \end{cases} \]
 and the functions $\tilde{\delta}^{i,n,b*}_\xi$ and $\tilde{\delta}^{i,n,a*}_\xi$ are given, for all $i \in \{1, \ldots, d\}$ and $n\in \{1,\ldots,N\}$, by
\[\tilde\delta^{i,n,b*}_\xi (z,p)
   =  \tilde\delta^{i,n,a*}_\xi (z,p) = \begin{cases}
 p + \frac{1}{\xi z}\log\left(1+\frac{\xi z}{k^{i,n}}\right)&\mbox{if }\xi>0\\
 p + \frac{1}{k^{i,n}}&\mbox{if }\xi=0.
 \end{cases} \]
Therefore, if we consider the quadratic approximation of the Hamiltonian functions based on their Taylor expansion around $p=0$ (see Remark \ref{remp0gen}), then \eqref{asymptbsymgen} and \eqref{asymptasymgen} become

\begin{eqnarray*}
	\breve\delta^{i,n,b}_t(z) &=& \begin{cases}
 \sqrt{\gamma} \left( q_{t-}^\intercal \Gamma e^i + \frac 12 z^i
     {e^i}^\intercal \Gamma e^i - \frac 1 \gamma {e^i}^\intercal D_+^{-\frac 12}
     \widehat A^+ D_+^{\frac 12} \mu \right) + \frac{c^{i,n,b}}{z}
 +\frac{1}{\gamma z}\log\left(1+\frac{\gamma z}{k^{i,n}}\right) &\mbox{ in Model A,}\\
 \sqrt{\gamma} \left( q_{t-}^\intercal \Gamma e^i + \frac 12 z^i
     {e^i}^\intercal \Gamma e^i - \frac 1 \gamma {e^i}^\intercal D_+^{-\frac 12}
     \widehat A^+ D_+^{\frac 12} \mu \right) + \frac{c^{i,n,b}}{z} +\frac{1}{k^{i,n}} &\mbox{ in  Model B.}
\end{cases}\\
	\breve\delta^{i,n,a}_t(z) &=& \begin{cases}
 \sqrt{\gamma} \left( q_{t-}^\intercal \Gamma e^i - \frac 12 z^i
     {e^i}^\intercal \Gamma e^i + \frac 1 \gamma {e^i}^\intercal D_+^{-\frac 12}
     \widehat A^+ D_+^{\frac 12} \mu \right) + \frac{c^{i,n,a}}{z}
 +\frac{1}{\gamma z}\log\left(1+\frac{\gamma z}{k^{i,n}}\right) &\mbox{ in Model A,}\\
 \sqrt{\gamma} \left( q_{t-}^\intercal \Gamma e^i - \frac 12 z^i
     {e^i}^\intercal \Gamma e^i + \frac 1 \gamma {e^i}^\intercal D_+^{-\frac 12}
     \widehat A^+ D_+^{\frac 12} \mu \right) + \frac{c^{i,n,a}}{z} +\frac{1}{k^{i,n}} &\mbox{ in  Model B.}
\end{cases}
	\end{eqnarray*}
where $\Gamma=D_+^{-\frac 1 2}(D_+^{\frac 12}\Sigma D_+^{\frac 12})^{\frac 12} D_+^{-\frac 1 2}$ and $$D_+=\textup{diag}\left(2 \sum_{n=1}^N A^{1,n} k^{1,n} \int_{\mathbb R ^*_+}C_\xi^{1,n}(z)  z \nu^{1,n}(dz) ,\ldots, 2 \sum_{n=1}^N A^{d,n} k^{d,n} \int_{\mathbb R ^*_+} C_\xi^{d,n}(z)  z  \nu^{d,n}(dz)\right).$$

\section{Numerical results}
\label{num_sec}
To assess the quality of our approximations, we consider a two-asset example for which we can approximate numerically the true function $\theta$ and the optimal quotes. By using a Euler scheme in dimension 3 (one dimension for time and two dimensions for inventory) it is indeed possible to approximate numerically the solution of Hamilton-Jacobi equations on a grid.\\

\subsection{Characteristics of our example with two assets}
\vspace{3mm}

We consider the following characteristics for the two assets:
\begin{itemize}
    \item Asset prices: $S^1_0 = S^2_0 = 100\ \textrm{\euro}$.
    \item Drifts: $\mu^1 = 0.1\ \textrm{\euro}\cdot \textrm{day}^{-1}$, $\mu^2 = -0.1\ \textrm{\euro}\cdot \textrm{day}^{-1}$.
    \item Volatilities: $\sigma^1 = 1.2\ \textrm{\euro}\cdot \textrm{day}^{-\frac 12}$, $\sigma^2 = 0.6\ \textrm{\euro}\cdot \textrm{day}^{-\frac 12}$.
    \item Correlation: $\rho = 0.5.$
\end{itemize}
This corresponds to a covariance matrix $\Sigma$ given by
$$\Sigma = \begin{pmatrix}(\sigma^1)^2 & \rho \sigma^1 \sigma^2 \\ \rho \sigma^1 \sigma^2 & (\sigma^2)^2  \end{pmatrix} = \begin{pmatrix}1.44 & 0.36 \\ 0.36 & 0.36  \end{pmatrix}.$$

We consider Model B\footnote{The results would be similar for Model A.} with time horizon $T = 7\ \textrm{days}$ and risk aversion parameter~$\gamma=8 \cdot 10^{-6}\ \textrm{\euro}^{-1}$.\\

We consider a framework with one tier only and no transaction costs.\\

The intensity functions are given for all $i \in \{1,2\}$ by:
$$\Lambda^{i,b}(\delta) = \Lambda^{i,a}(\delta) = \lambda_{RFQ} \frac 1{1 + e^{\alpha_\Lambda + \beta_\Lambda \delta}},$$
with $\lambda_{RFQ} = 30\ \textrm{day}^{-1}$, $\alpha_\Lambda = 0.7$, and $\beta_\Lambda = 30\ \textrm{\euro}^{-1}$. This corresponds to $30$ requests per day, a probability of $\frac 1{1 + e^{0.7}}\simeq 33 \% $ to trade when the answered quote is the reference price and a probability of $\frac 1{1+e^{0.4}} \simeq 40\% $ to trade when the answered quote is the reference price improved by 1 cent.\\

Request sizes are distributed according to a Gamma distribution $\Gamma(\alpha_\mu, \beta_\mu)$ with $\alpha_\mu = 4$ and $\beta_\mu = 4 \cdot 10^{-4}.$ This corresponds to an average request size of $10000$ assets (i.e. approximately $1000000 \textrm{\euro}$) and a standard deviation equal to half the average.\\

\subsection{Value function and optimal quotes}
\vspace{3mm}
In order to discretize the problem, we first approximate the Gamma distribution of sizes with $4$ sizes: $z^1 = 6250$, $z^2 = 12500$, $z^3 = 18750$, and $z^4 = 25000$ assets -- thereafter refered to by very small, small, large, and very large size -- with respective probability $p^1 = 0.534$, $p^2 = 0.350$, $p^3 = 0.097$ and $p^4 = 0.019$. We impose risk limits $Q^1 = 75000$ and $Q^2= 300000$, i.e. no trade that would result in an inventory $q^1$ for asset~1 such that $|q^1|>75000$ is accepted, and similarly no trade that would result in an inventory $q^2$ for asset~2 such that $|q^2|>300000$ is accepted.\\

The solution $\theta$ to Eq. \eqref{HJfeat} with terminal condition \eqref{sec5:thetagenCT} can then be approximated numerically using a monotone implicit Euler scheme on a grid of size $101\times25\times 97$ ($101$ points in time, $25$ points for the inventory of asset $1$, and $97$ points for the inventory of asset 2).\\

Because we are mainly interested in asymptotic quotes, it is important to check that the time horizon chosen is sufficienty long. For that purpose, we plot in Figure \ref{conv_deltas_BEGV} the optimal bid quotes for asset 1 and asset 2 at time $t=0$ for different values of the inventory. We see that the asymptotic regime is clearly reached and, from now on, we will only consider the value function and the optimal quotes at time $t=0$.\\

\begin{figure}[!h]\centering
\includegraphics[width=\textwidth]{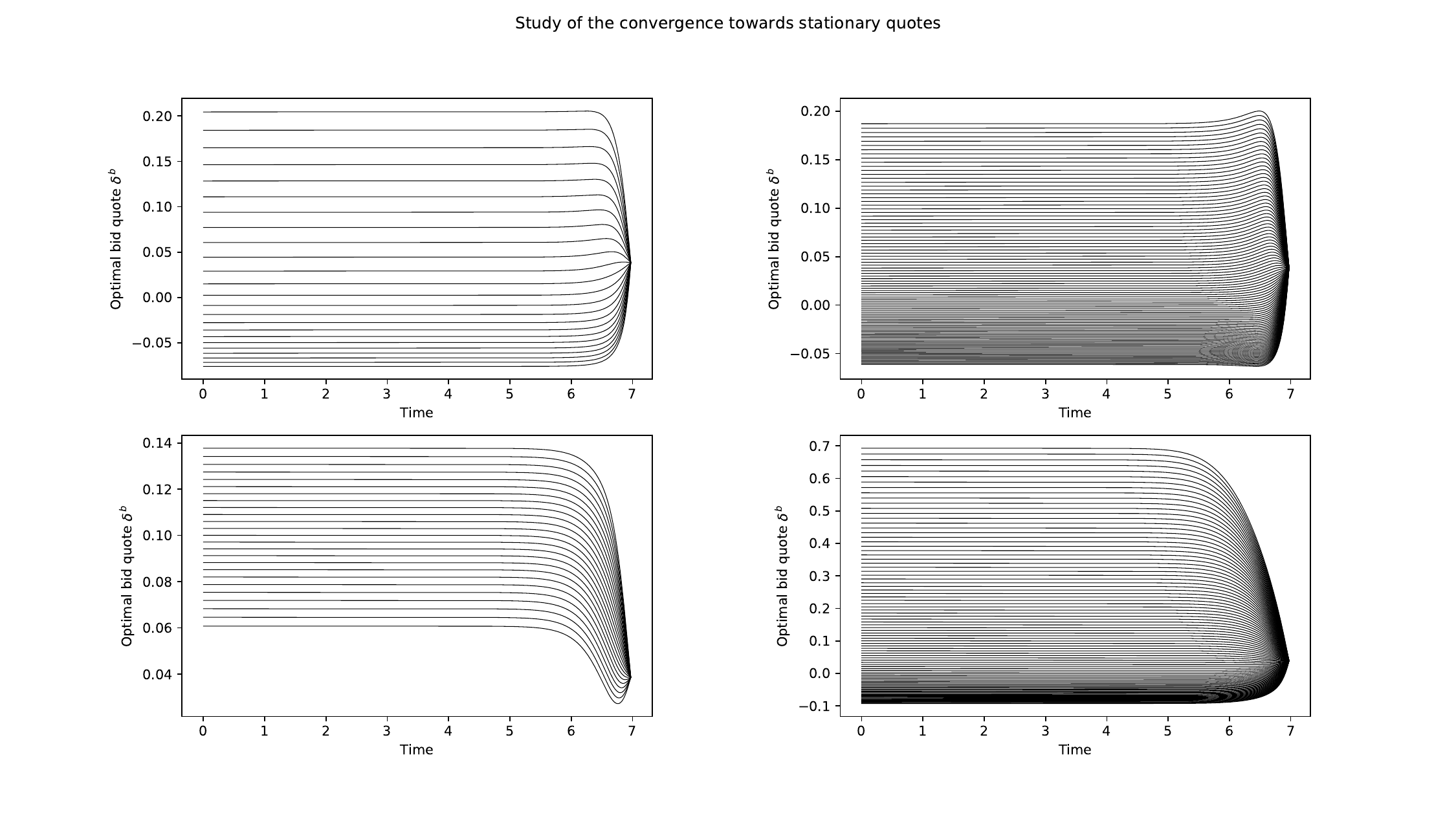}\\
\caption{Optimal bid quotes as a function of time for different values of the inventory (very small trades) -- top left: bid quotes of asset 1 for different values of inventory $q^2$ ($q^1=0$), top right: bid quotes of asset 1 for different values of inventory $q^1$ ($q^2=0$), bottom left: bid quotes of asset 2 for different values of inventory $q^2$ ($q^1=0$), bottom right: bid quotes of asset 2 for different values of inventory $q^1$ ($q^2=0$).}\label{conv_deltas_BEGV}
\end{figure}

\begin{figure}[!h]\centering
\includegraphics[width=0.97\textwidth]{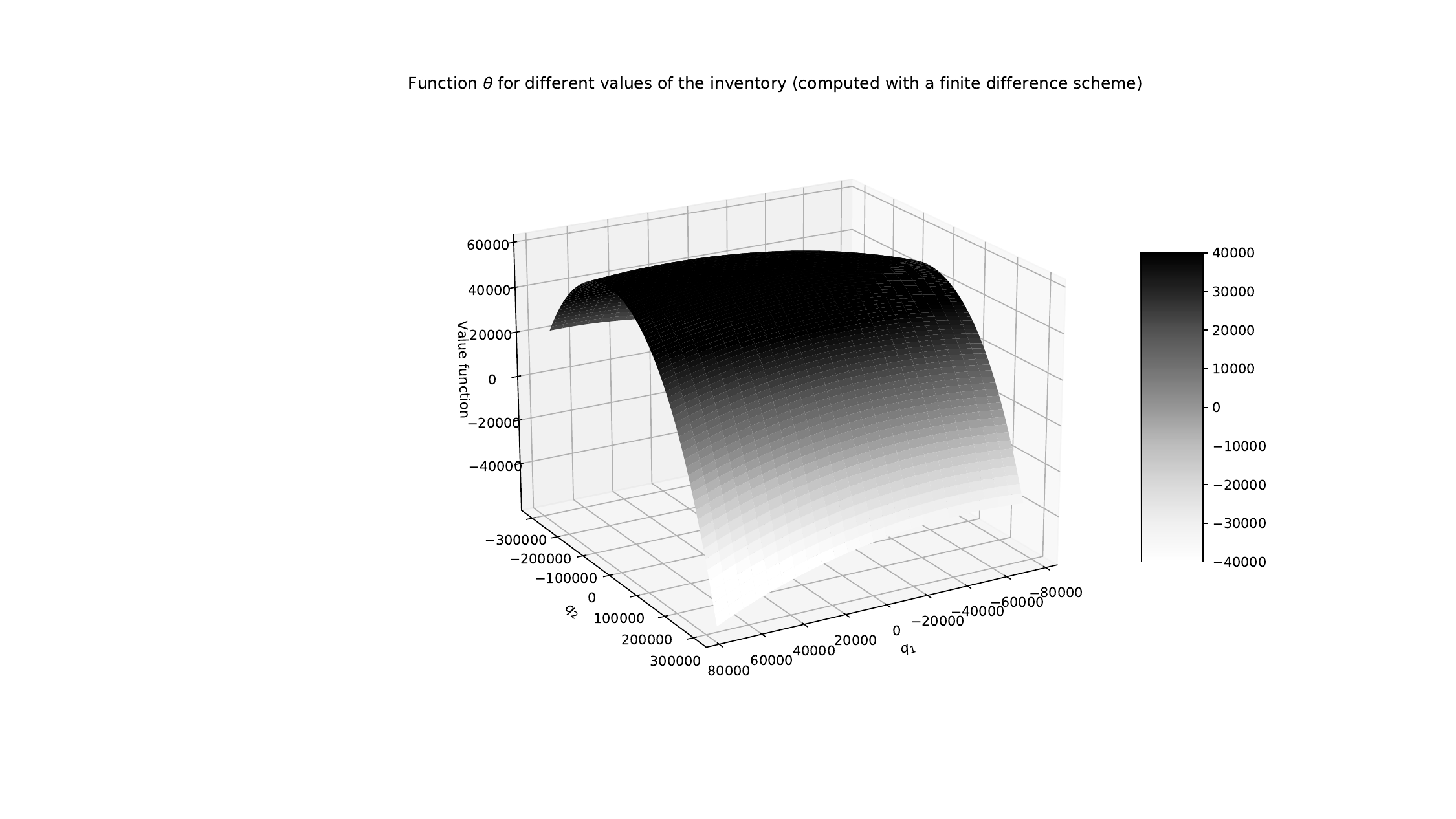}\\
\caption{Function $\theta$ at time $t=0$ for different values of the inventory.}\label{theta_3d_BEGV}
\end{figure}

The numerical approximation of the value function $\theta$ (at time $t=0$) is plotted in Figure \ref{theta_3d_BEGV}. The shape of the function $\theta$ is as expected given the risk penalty, the positive drift in the prices of asset 1 and the negative drift in the prices of asset 2. The associated bid quotes are plotted in Figures \ref{delta_b_3d_asset_1_BEGV} and \ref{delta_b_3d_asset_2_BEGV} respectively. The shape of the quote surfaces is as expected given the positive correlation coefficient (see \cite{bergault2019size} or \cite{gueant2017optimal} for a deeper discussion about the effect of the different parameters).\\

\begin{figure}[!h]\centering
\includegraphics[width=\textwidth]{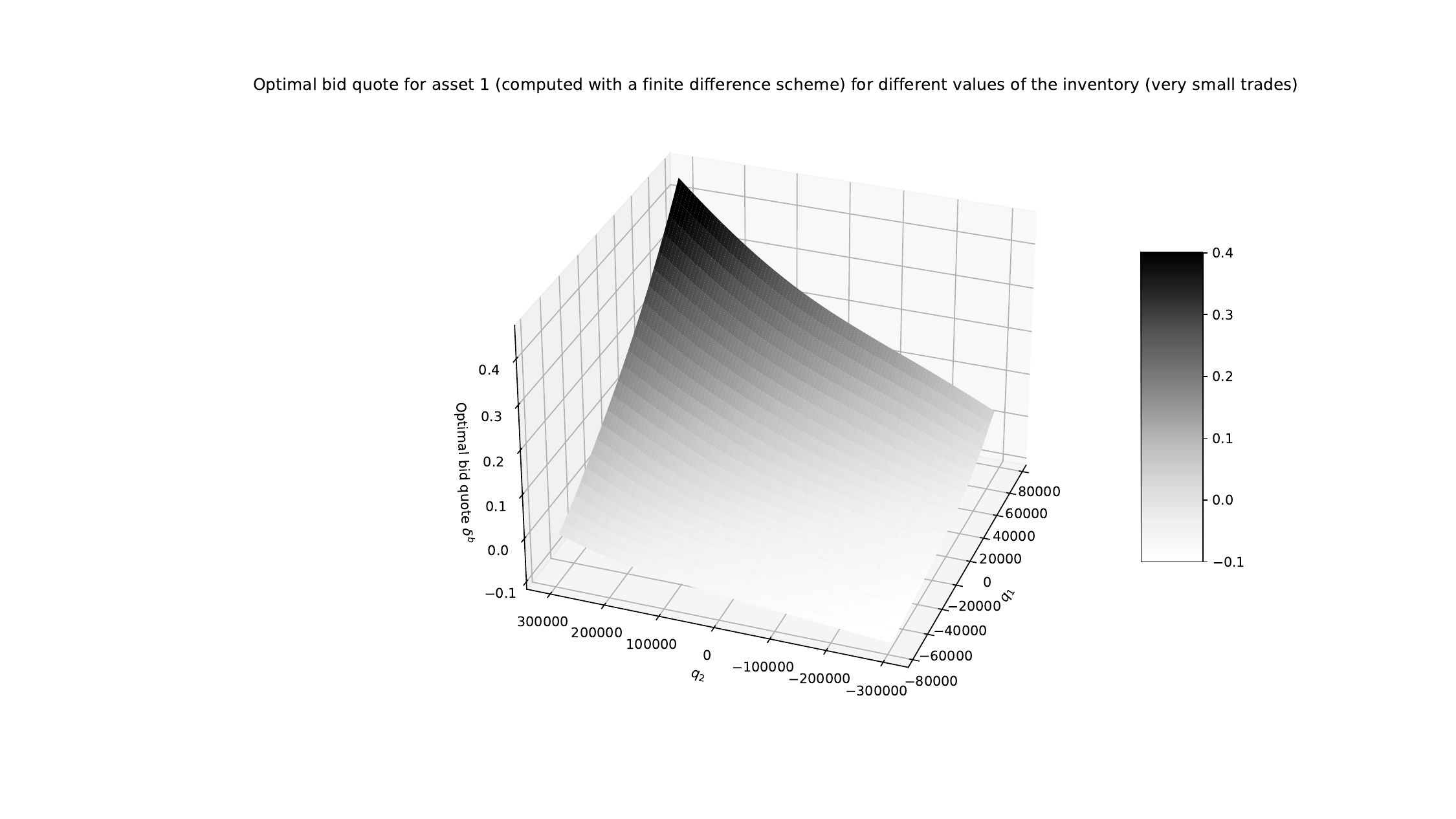}\\
\caption{Optimal bid quote at $t=0$ for asset 1 as a function of the inventory (very small trades).}\label{delta_b_3d_asset_1_BEGV}
\end{figure}
\newpage

\begin{figure}[!h]\centering
\includegraphics[width=\textwidth]{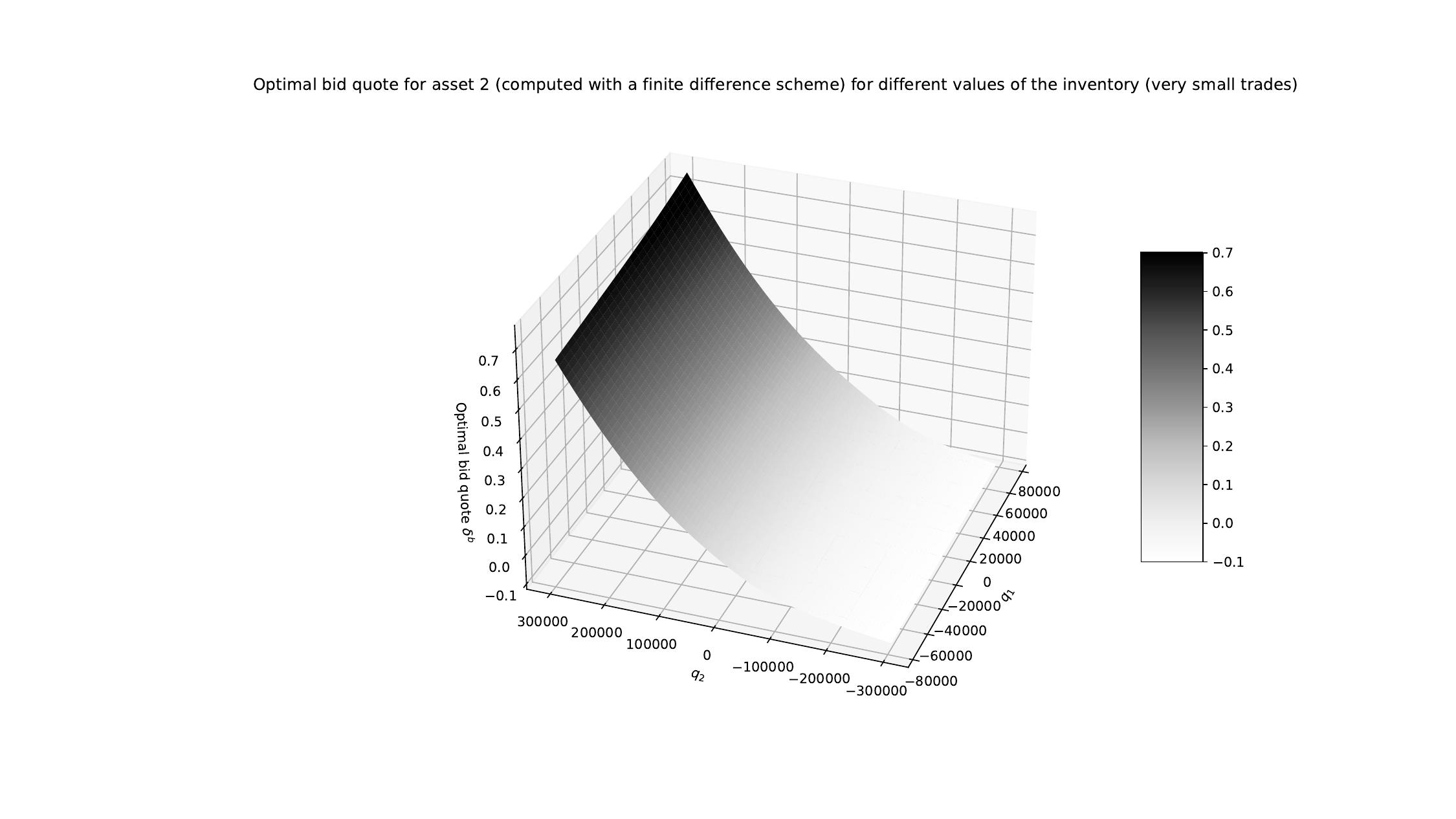}\\
\caption{Optimal bid quote at $t=0$ for asset 2 as a function of the inventory (very small trades).}\label{delta_b_3d_asset_2_BEGV}
\end{figure}

\subsection{Comparison with closed-form approximations}
\vspace{3mm}
We now move on to our closed-form approximations. We first plot in Figure \ref{theta_hat_3d_BEGV} the closed-form approximation $\check{\theta}$ given by Proposition \ref{prop:GENansatz}.\\

\begin{figure}[!h]\centering
\includegraphics[width=\textwidth]{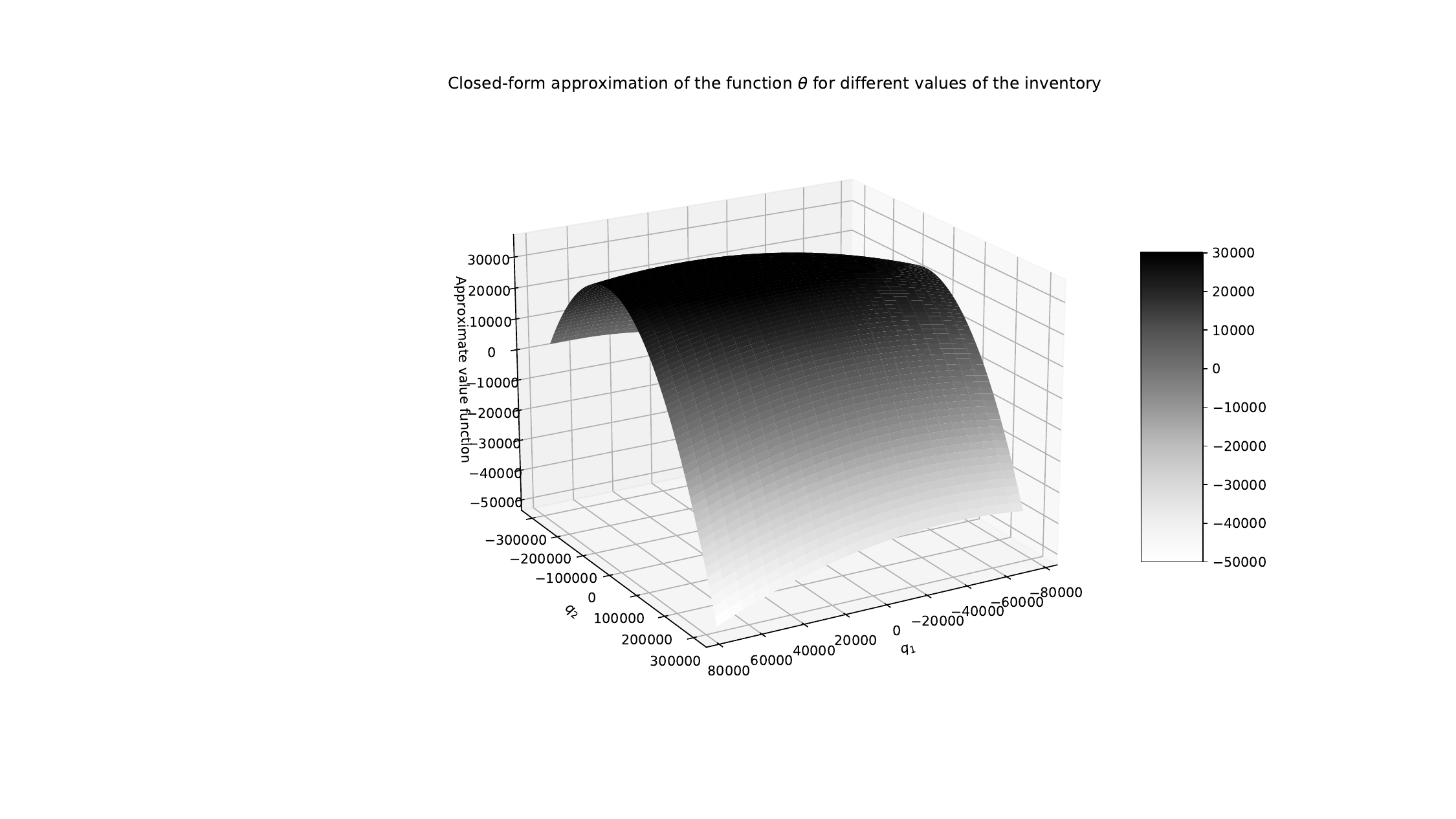}\\
\caption{Function $\check{\theta}$ at $t=0$ for different values of the inventory.}\label{theta_hat_3d_BEGV}
\end{figure}

We clearly see that, in spite of differences in level between the true value function aproximated numerically and the closed-form approximation, the shape is the same. Therefore, the finite differences involved in the computation of the associated quotes should be similar. This is roughly confirmed in Figures \ref{delta_hat_b_3d_asset_1_BEGV} and \ref{delta_hat_b_3d_asset_2_BEGV} that are the closed-form counterparts of Figures \ref{delta_b_3d_asset_1_BEGV} and \ref{delta_b_3d_asset_2_BEGV}.\\

\begin{figure}[!h]\centering
\includegraphics[width=\textwidth]{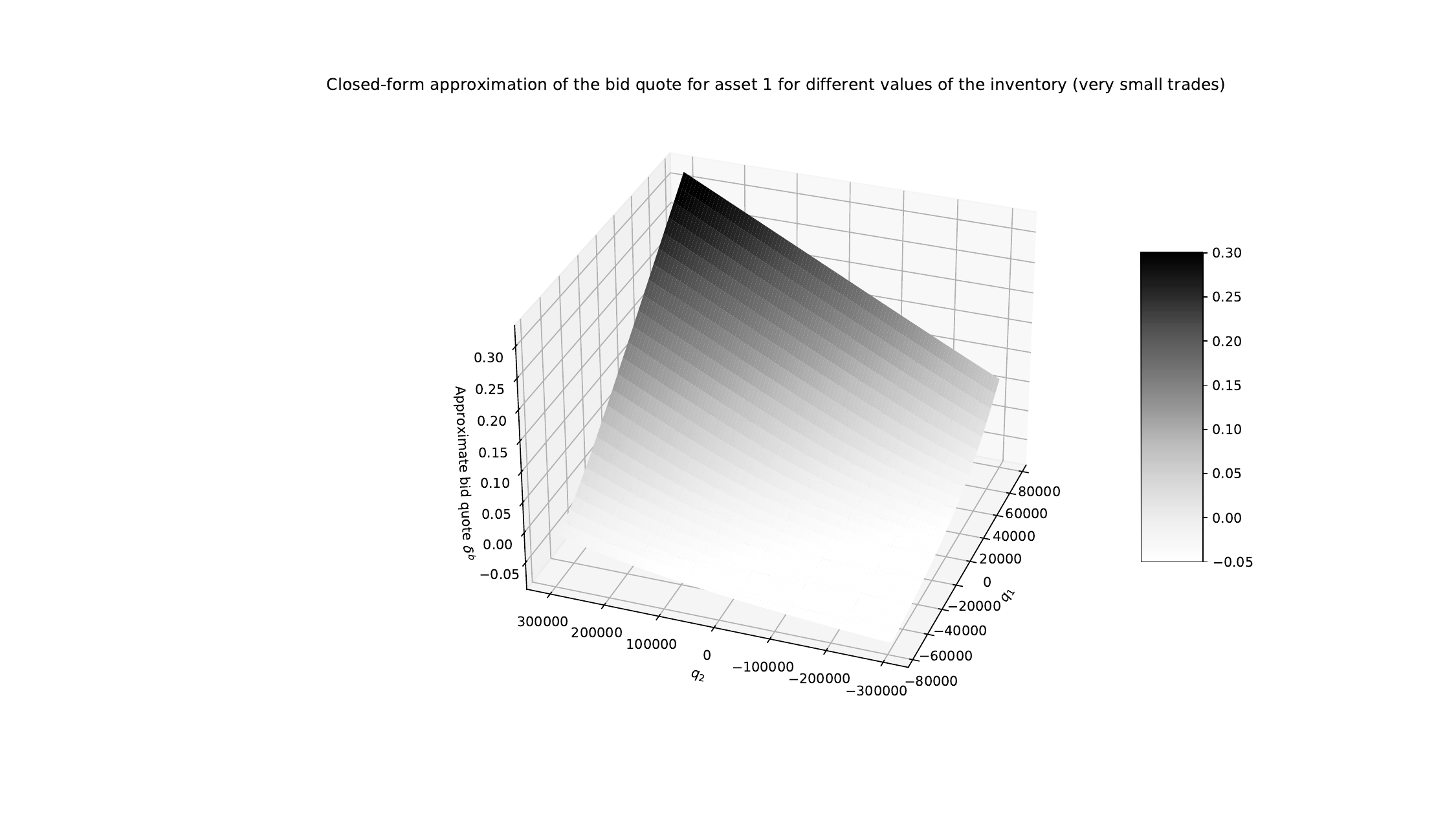}\\
\caption{Closed-form approximation for the optimal bid quote at $t=0$ for asset 1 as a function of the inventory (very small trades).}\label{delta_hat_b_3d_asset_1_BEGV}
\end{figure}

\begin{figure}[!h]\centering
\includegraphics[width=\textwidth]{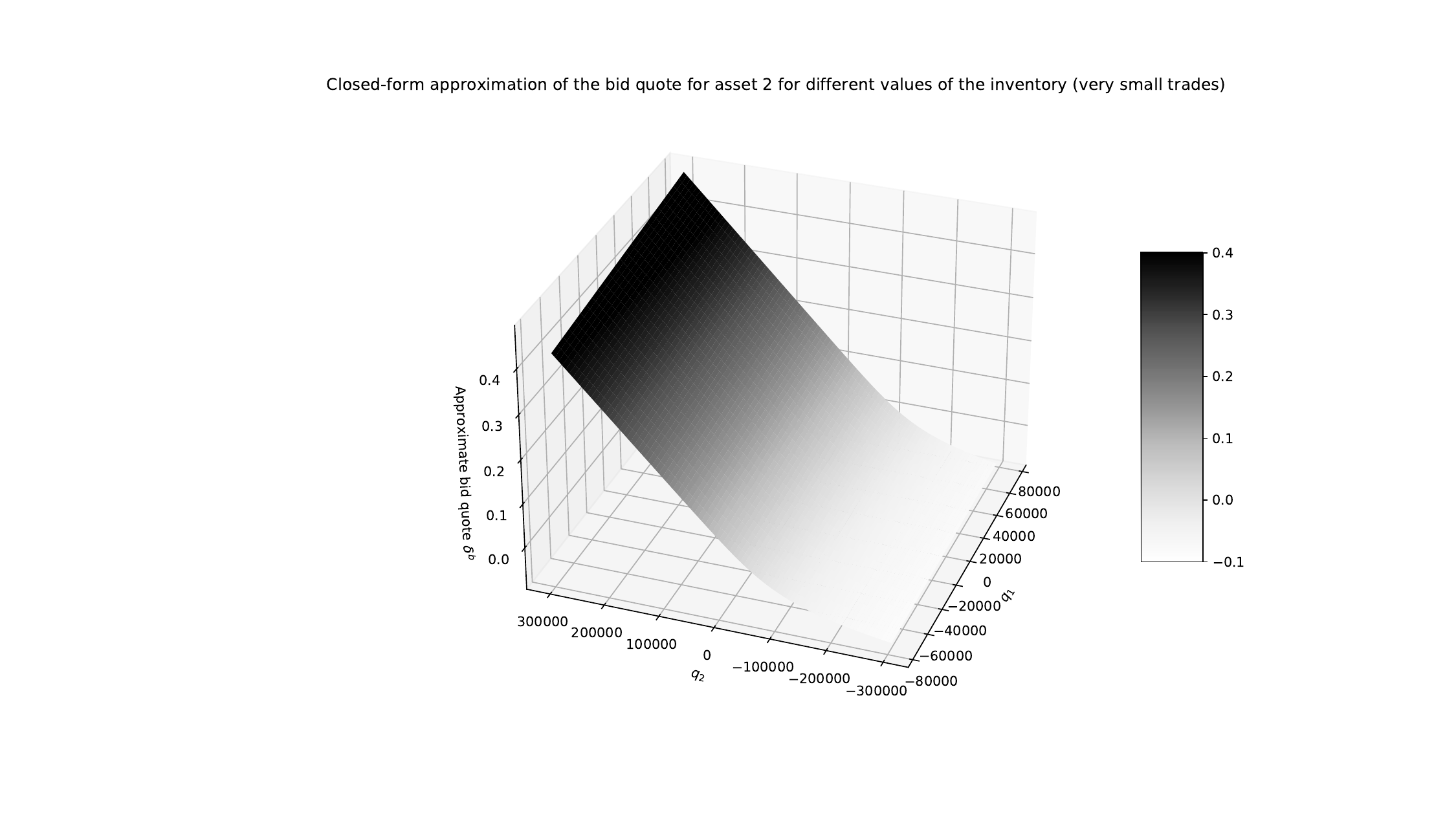}\\
\caption{Closed-form approximation for the optimal bid quote at $t=0$ for asset 2 as a function of the inventory (very small trades).}\label{delta_hat_b_3d_asset_2_BEGV}
\end{figure}

In order to assess more precisely the quality of our closed-form approximations, we plot in Figures \ref{deltas_comp_asset_1_q_1_BEGV}, \ref{deltas_comp_asset_1_q_2_BEGV}, \ref{deltas_comp_asset_2_q_1_BEGV}, and \ref{deltas_comp_asset_2_q_2_BEGV} the functions
$$q^1 \mapsto \bar \delta^{1,b}(q^1,0, z^k),\ k \in \{1,\ldots,4\} \qquad q^1 \mapsto \hat \delta^{1,b}(q^1,0, z^k),\ k \in \{1,\ldots,4\} $$
$$q^2 \mapsto \bar \delta^{1,b}(0,q^2, z^k),\ k \in \{1,\ldots,4\} \qquad q^2 \mapsto \hat \delta^{1,b}(0,q^2, z^k),\ k \in \{1,\ldots,4\}$$
$$q^1 \mapsto \bar \delta^{2,b}(q^1,0, z^k),\ k \in \{1,\ldots,4\} \qquad q^1 \mapsto \hat \delta^{2,b}(q^1,0, z^k),\ k \in \{1,\ldots,4\}$$
$$q^2 \mapsto \bar\delta^{2,b}(0,q^2, z^k),\ k \in \{1,\ldots,4\} \qquad q^2 \mapsto \hat\delta^{2,b}(0,q^2, z^k),\ k \in \{1,\ldots,4\}$$
where $\bar \delta^{i,b}$ is the optimal bid quote for asset $i$ as a function of time, inventory, and size of request and $\hat \delta^{i,b}$ is the closed-form approximation of the optimal bid quote for asset $i$ as a function of inventory and size of request.\\

We clearly see that our closed-form approximations are close to the true optimal quotes, except for large values of the inventory in asset 2.\\

\begin{figure}[!h]\centering
\includegraphics[width=\textwidth]{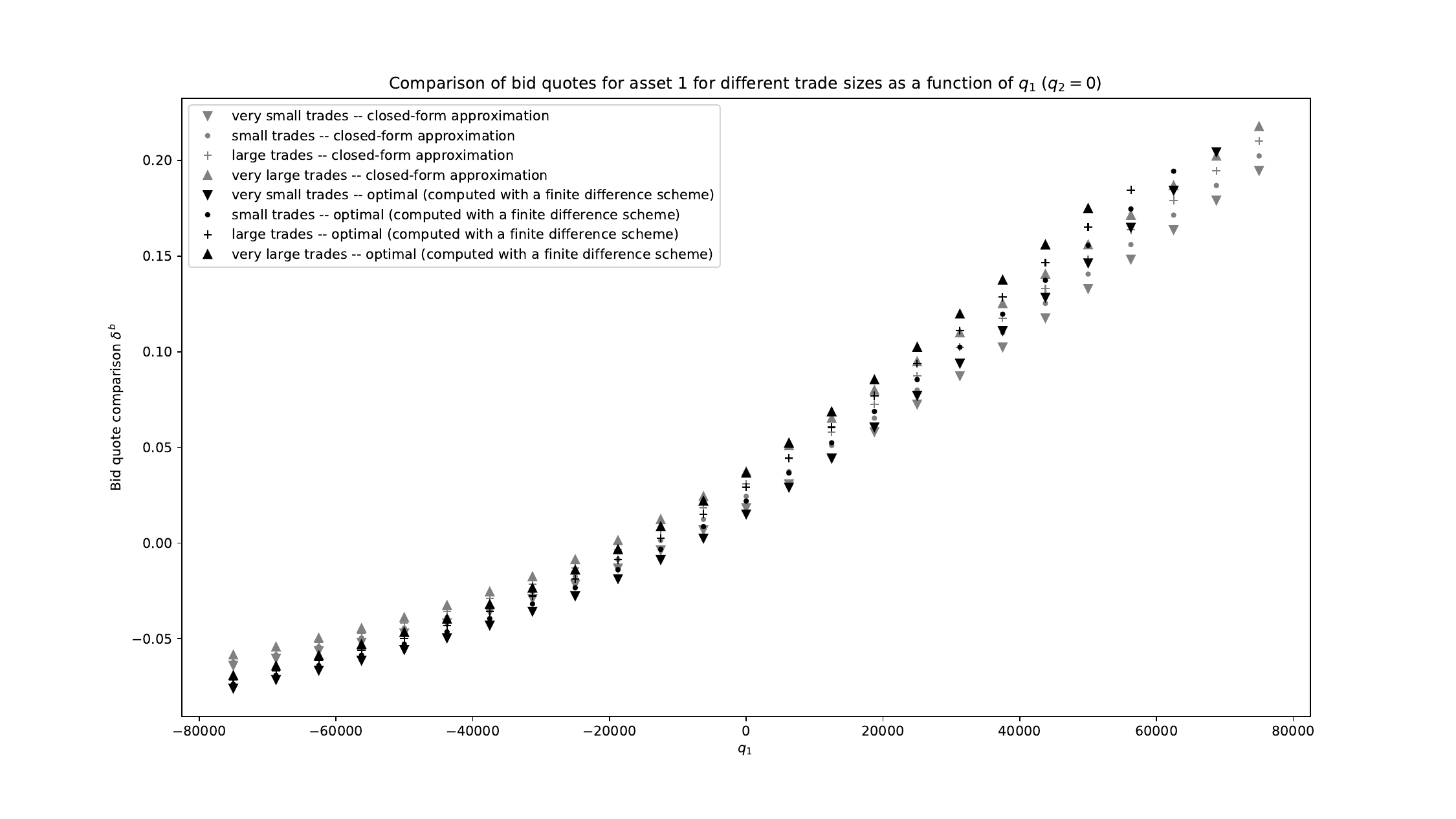}\\
\caption{Comparison between optimal bid quote for asset 1 and its closed-form approximation for different trade sizes as a function of $q^1$ ($q^2=0$).}\label{deltas_comp_asset_1_q_1_BEGV}
\end{figure}

\begin{figure}[!h]\centering
\includegraphics[width=\textwidth]{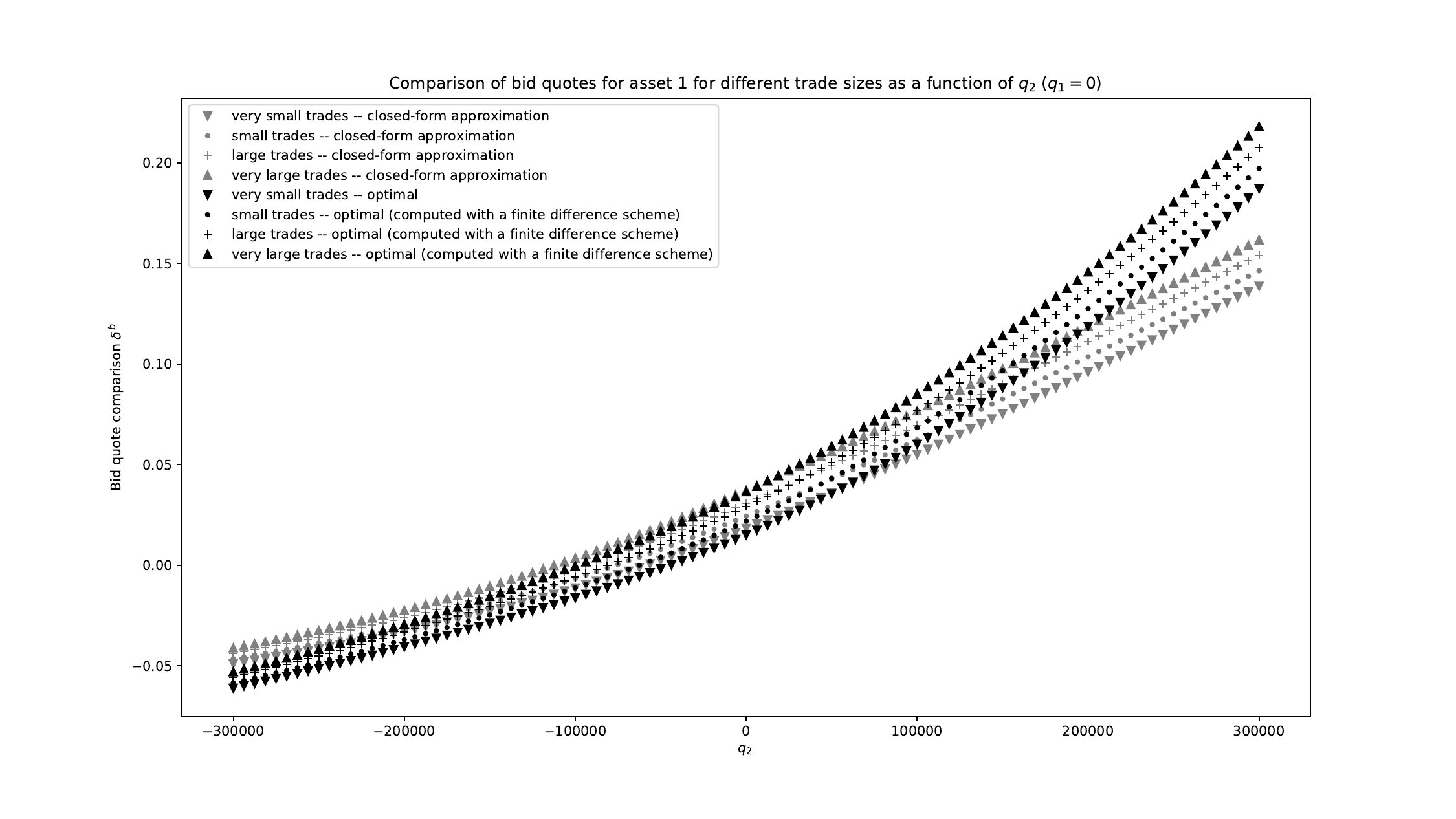}\\
\caption{Comparison between optimal bid quote for asset 1 and its closed-form approximation for different trade sizes as a function of $q^2$ ($q^1=0$).}\label{deltas_comp_asset_1_q_2_BEGV}
\end{figure}

\begin{figure}[!h]\centering
\includegraphics[width=\textwidth]{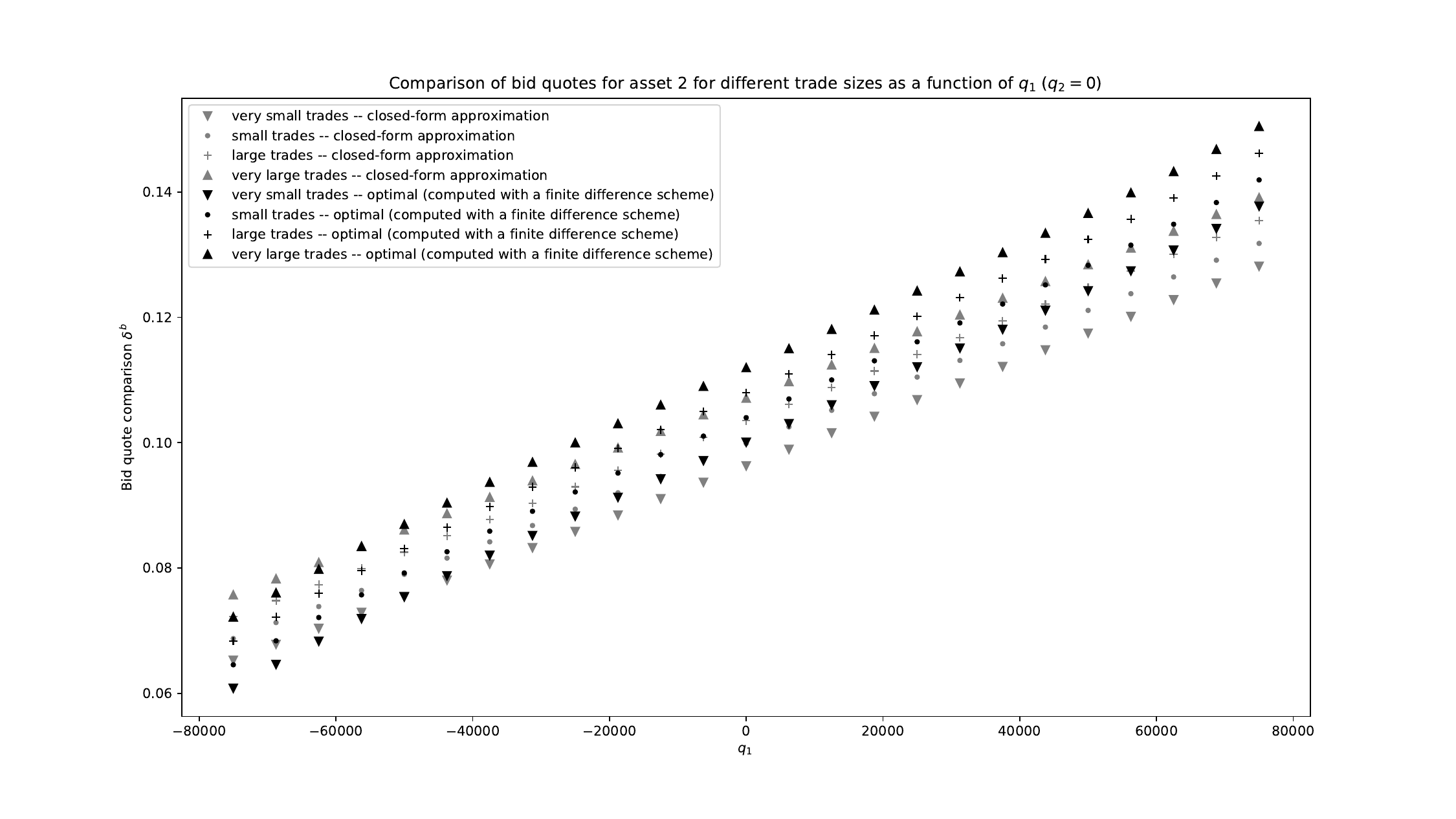}\\
\caption{Comparison between optimal bid quote for asset 2 and its closed-form approximation for different trade sizes as a function of $q^1$ ($q^2=0$).}\label{deltas_comp_asset_2_q_1_BEGV}
\end{figure}

\begin{figure}[!h]\centering
\includegraphics[width=\textwidth]{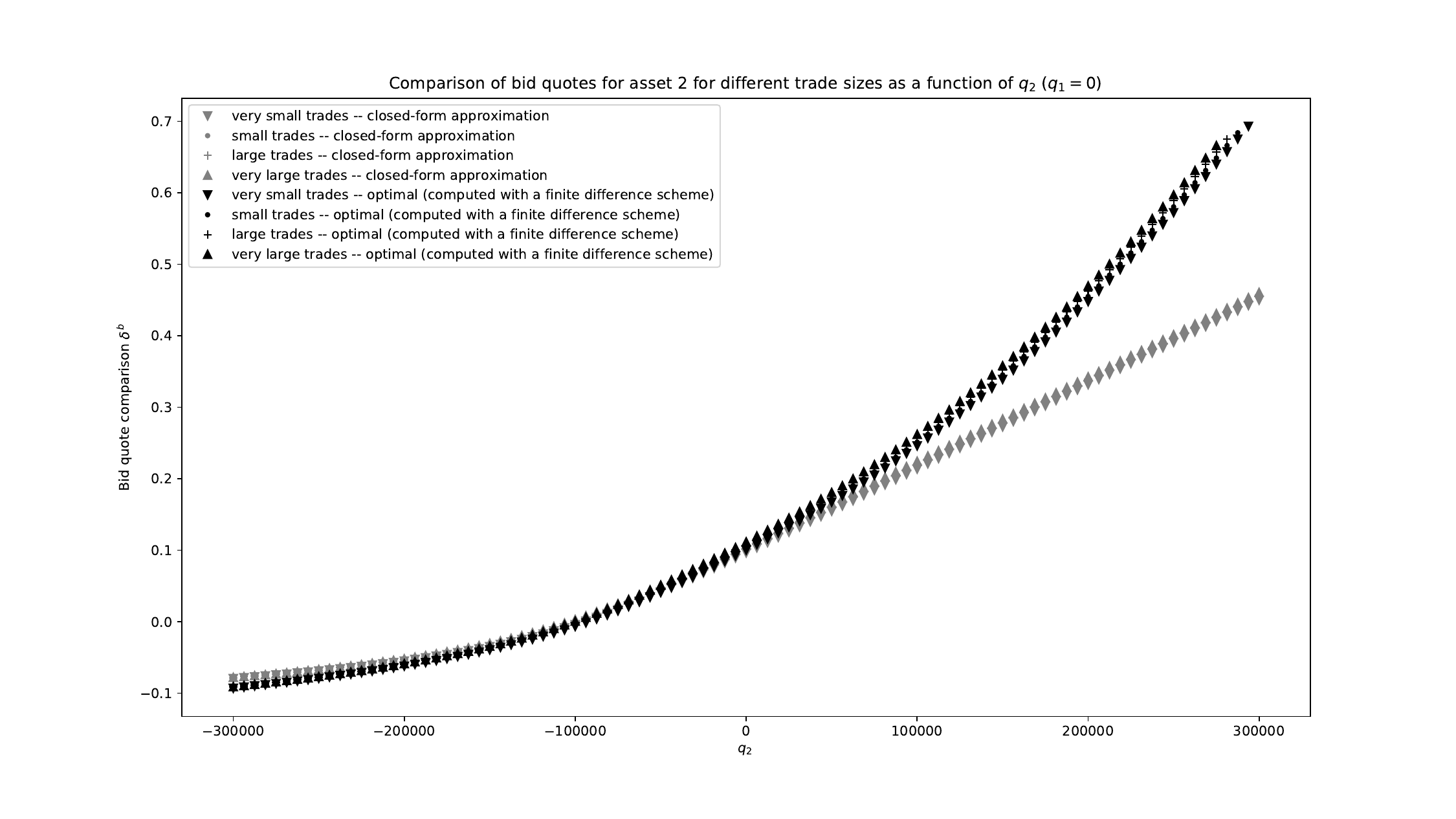}\\
\caption{Comparison between optimal bid quote for asset 2 and its closed-form approximation for different trade sizes as a function of $q^2$ ($q^1=0$).}\label{deltas_comp_asset_2_q_2_BEGV}
\end{figure}

In order to confirm the quality of our closed-form approximations, we compare the performance of a market maker, when quoting the true optimal quotes versus their closed-form approximations. The respective distributions of PnL after 4000 Monte-Carlo simulations are plotted in Figures \ref{pnl_distrib_optimal_BEGV} and \ref{pnl_distrib_approx_BEGV}.

\begin{figure}[!h]\centering
\includegraphics[width=0.85\textwidth]{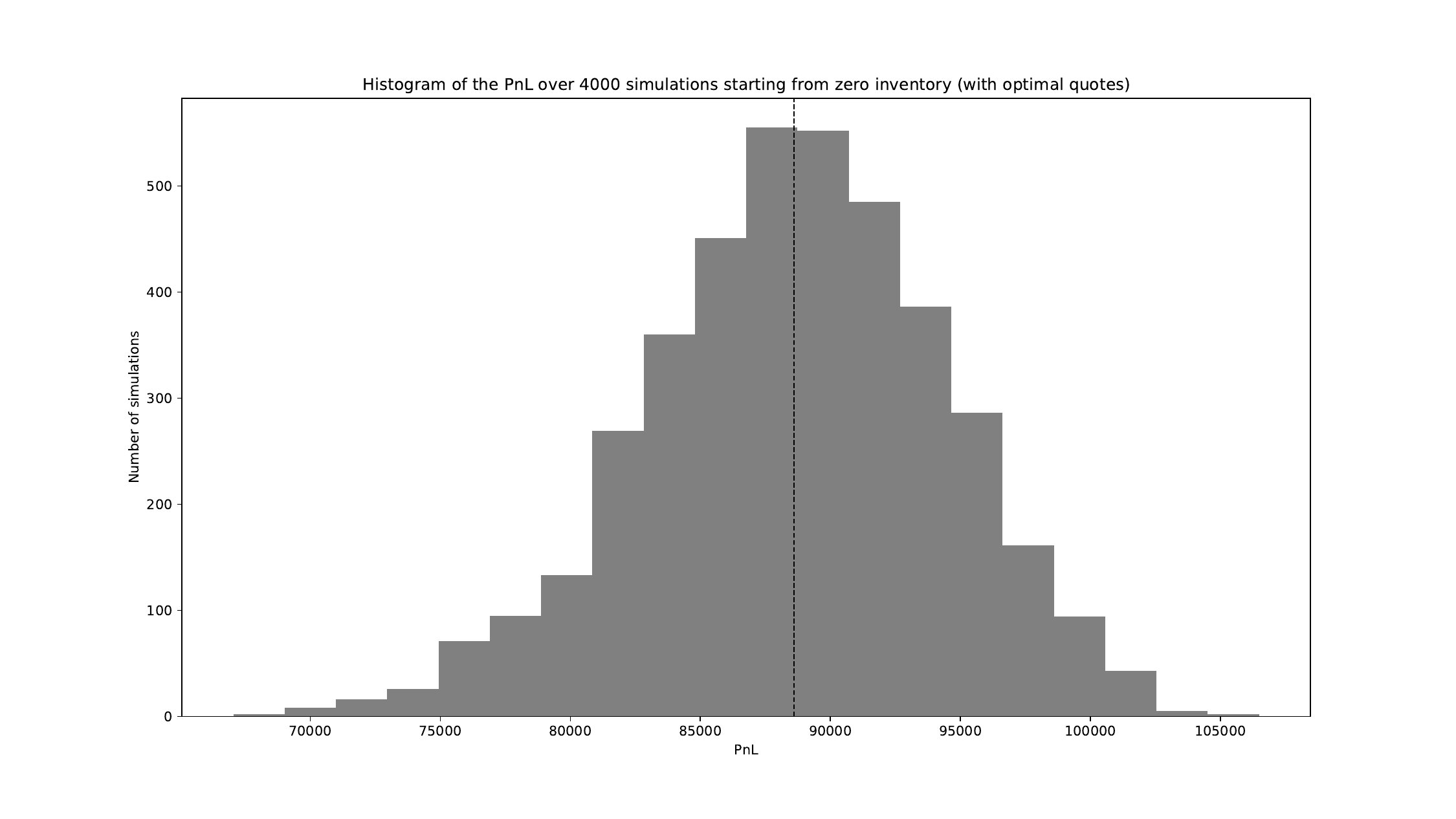}\\
\caption{Distribution of the PnL of a market maker using the optimal quotes.}\label{pnl_distrib_optimal_BEGV}
\end{figure}

\begin{figure}[!h]\centering
\includegraphics[width=0.85\textwidth]{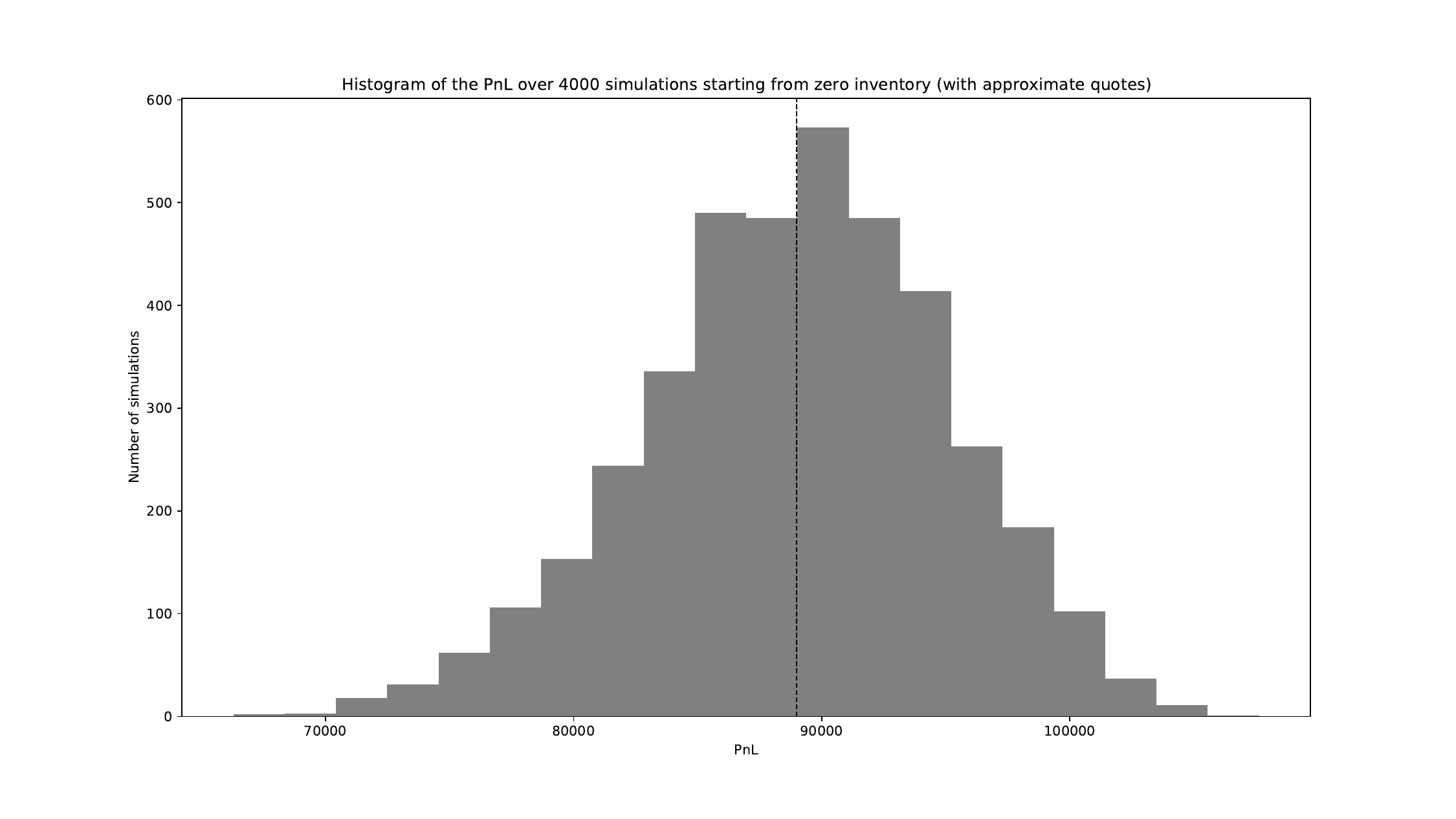}\\
\caption{Distribution of the PnL of a market maker using the closed-form approximations.}\label{pnl_distrib_approx_BEGV}
\end{figure}

When using the optimal quotes, the market maker gets an average PnL of $ 88600\textrm{\euro}$  with a standard deviation of $86900\textrm{\euro}$. When using the closed-form approximation, the performance is very similar, as she gets an average PnL of $ 89000 \textrm{\euro}$ with a standard deviation of $87500\textrm{\euro}$.\\

These results are really satisfying in terms of performance. We see that, although the closed-form approximation of optimal quotes may be inaccurate for large values of the inventory, such large inventory are seldom reached and therefore the gap between quotes does not really impact the distribution of the PnL.\\

We believe that what we observe here in this two-asset example is general. In particular, the results in higher dimensions should be just as good.

\section*{Conclusion}

We proposed closed-form approximations for the value functions associated with many multi-asset extensions of the market making models available in the literature. These closed-form approximations have been obtained through the ``approximation'' of a Hamilton-Jacobi equation by another Hamilton-Jacobi equation that can be simplified into a Riccati equation and two linear ordinary differential equations, all solvable in closed form. These closed-form approximations can be used for various purposes, in particular to design quoting strategies through a greedy approach. The resulting closed-form approximations of the optimal quotes generalize those obtained by Guéant, Lehalle, and Fernandez-Tapia in \cite{gueant2013dealing} to a general framework suitable for practical use.\\

\section*{Acknowledgements}

The authors would like to thank Bastien Baldacci (Ecole Polytechnique) and Iuliia Manziuk (Ecole Polytechnique) for their careful reading of an initial version of the paper. Olivier Guéant would like to thank the Research Initiative ``Nouveaux traitements pour les données lacunaires issues des activités de crédit'' financed by BNP Paribas under the aegis of the Europlace Institute of Finance for its support. A special thank goes to Laurent Carlier (BNP Paribas) for his vivid interest in academic questions around market making throughout the years.\\

\begin{appendix}
\section{On the construction of the processes $N^{i,b}$ and $N^{i,a}$}
\label{ConstrucN}
Let us consider a new filtered probability space $\big(\Omega,\mathcal{F},(\mathcal{F}_t)_{t\in \mathbb{R}_{+}},\tilde{\mathbb{P}}\big)$. For the sake of simplicity, assume that there is only one asset with size of transactions denoted by $z$ and risk limit $Q$ (the generalization is straightforward). Let us assume that the reference price of that asset is driven by a Brownian motion $W$ as in Section \ref{baseModel}. Let us introduce $\bar{N}^b$ and $\bar{N}^a$ two independent Poisson processes of intensity $1$, independent of~$W$. Let $N^b$ and $N^a$ be two processes, starting at 0, solutions of the coupled stochastic differential equation:
\begin{align*}
    dN^b_t =  \mathds{1}_{\left\{zN^b_{t-} - zN^a_{t-} + z \le Q\right\}} d\bar{N}^b_t,\\
    dN^a_t =  \mathds{1}_{\left\{zN^b_{t-} - zN^a_{t-} -z \ge  -Q\right\}} d\bar{N}^a_t.
\end{align*}
Then, under $\tilde{\mathbb{P}},$ $N^b$ and $N^a$ are two point processes with respective intensities $$\lambda^b_t = \mathds{1}_{\left\{q_{t-} + z \le Q \right\}} \quad \text{and} \quad \lambda^a_t =  \mathds{1}_{\left\{q_{t-} - z \ge -Q\right\}},$$
where $q_t = z N^b_t -  z N^a_t.$ For each $\delta \in \mathcal{A}$, we introduce the probability measure $\tilde{\mathbb{P}}^\delta$ given by the Radon-Nikodym derivative
\begin{equation}
\frac{d\tilde{\mathbb{P}}^\delta}{d\tilde{\mathbb{P}}} \Big|_{\mathcal{F}_t} = L_t^{\delta},
\end{equation}
where $\left(L_t^{\delta} \right)_{t \in \mathbb R^+}$ is the unique solution of the stochastic differential equation
\begin{equation}
    dL_t^{\delta} = L_{t-}^{\delta} \left( \left(\Lambda^b(\delta^b_t)-1 \right)d\tilde{N}^b_t+  \left(\Lambda^a(\delta^a_t)-1 \right)d\tilde{N}^a_t \right),\nonumber
\end{equation}
with $L_0^{\delta} = 1$, where $\tilde{N}^b$ and $\tilde{N}^a$ are the compensated processes associated with $N^b$ and $N^a$, respectively.\\

We then know from the Brémaud-Jacod version of Girsanov theorem (see \cite{bj}) that under $\tilde{\mathbb{P}}^\delta$, the jump processes $N^{b}$ and $N^{a}$ have respective intensities
$$\lambda^{\delta,b}_t  = \Lambda^b(\delta^b_t) \mathds{1}_{\left\{q_{t-} + z \le Q\right\}}  \quad \text{and} \quad \lambda^{\delta,a}_t  = \Lambda^a(\delta^a_t) \mathds{1}_{\left\{q_{t-} - z\ge -Q\right\}} $$ as in Section \ref{baseModel} of the paper. Since $W$ is still a Brownian motion under $\tilde{\mathbb{P}}^\delta$, our optimal control problem can be seen as the choice of an optimal probability measure $\tilde{\mathbb{P}}^\delta$.\\

\end{appendix}



\end{document}